\documentclass[sigconf]{acmart}

\usepackage{bm}
\usepackage{subfig}
\usepackage{natbib}
\usepackage[utf8]{inputenc} 
\usepackage[T1]{fontenc}    
\usepackage{booktabs}       
\usepackage{amsfonts}       
\usepackage{nicefrac}       
\usepackage{microtype}      
\usepackage{caption}
\usepackage{url}            
\usepackage{multirow}
\usepackage{enumitem}
\usepackage[ruled,vlined]{algorithm2e}

\newcommand{\ie}{\textit{i.e.}}
\newcommand{\eg}{\textit{e.g.}}
\newcommand{\improv}{\scalebox{1.2}{$\blacktriangle$\%} }

\newtheorem{remark}{Remark}

\AtBeginDocument{%
  \providecommand\BibTeX{{%
    \normalfont B\kern-0.5em{\scshape i\kern-0.25em b}\kern-0.8em\TeX}}}

\setcopyright{acmcopyright}
\copyrightyear{2024}
\acmYear{2024}
\acmDOI{XXXXXXX.XXXXXXX}

\acmPrice{15.00}
\acmISBN{978-1-4503-XXXX-X/18/06}

\begin{document}

\title{
    Understanding the Role of Cross-Entropy Loss in Fairly Evaluating Large Language Model-based Recommendation
}


\author{Cong Xu}
\authornote{Equal contribution}
\affiliation{%
  \institution{East China Normal University}
  \city{Shanghai}
  \country{China}
}
\email{congxueric@gmail.com}

\author{Zhangchi Zhu}
\authornotemark[1]
\affiliation{%
  \institution{East China Normal University}
  \city{Shanghai}
  \country{China}
}
\email{zczhu@stu.ecnu.edu.cn}

\author{Jun Wang}
\affiliation{%
  \institution{East China Normal University}
  \city{Shanghai}
  \country{China}
}
\email{wongjun@gmail.com}

\author{Jianyong Wang}
\affiliation{%
  \institution{Tsinghua University}
  \city{Beijing}
  \country{China}
}
\email{jianyong@tsinghua.edu.cn}

\author{Wei Zhang}
\affiliation{%
  \institution{East China Normal University}
  \city{Shanghai}
  \country{China}
}
\email{zhangwei.thu2011@gmail.com}

\begin{abstract}

Large language models (LLMs) have gained much attention in the recommendation community;
some studies have observed that LLMs, fine-tuned by the cross-entropy loss with a full softmax,
could achieve state-of-the-art performance already.
However, these claims are drawn from unobjective and unfair comparisons.
In view of the substantial quantity of items in reality,
conventional recommenders typically adopt a pointwise/pairwise loss function instead for training.
This substitute however causes severe performance degradation,
leading to under-estimation of conventional methods and over-confidence in the ranking capability of LLMs.

In this work,
we theoretically justify the superiority of cross-entropy, 
and showcase that it can be adequately replaced by some elementary approximations with certain necessary modifications.
The remarkable results across three public datasets corroborate that even in a practical sense, 
existing LLM-based methods are not as effective as claimed for next-item recommendation.
We hope that these theoretical understandings in conjunction with the empirical results
will facilitate an objective evaluation of LLM-based recommendation in the future. 
Our code is available at \url{https://github.com/MTandHJ/CE-SCE-LLMRec}.

\end{abstract}





\maketitle

\section{Introduction}

\begin{figure}
	\centering
	\includegraphics[width=0.47\textwidth]{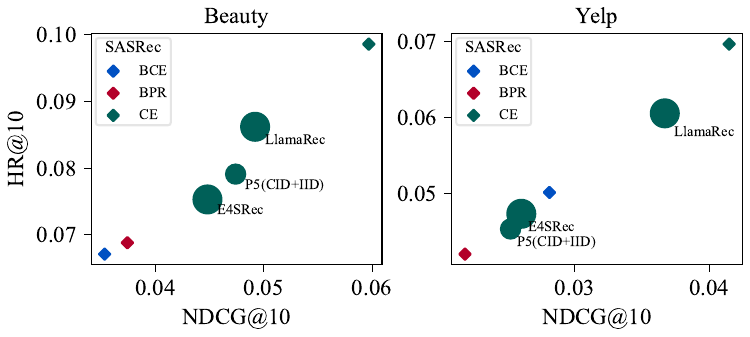}
	\caption{
        Recommendation performance comparisons.
        The marker size depicts the number of model parameters:
        60M for P5 (CID + IID) \cite{LLM4Rec:CIDIID:Hua2023}, 7B for LlamaRec \cite{LLM4Rec:LLaMARec:Yue2023} and E4SRec \cite{LLM4Rec:E4SRec:Li:2023}, 
        and merely $\le 1$M for SASRec \cite{Seq:SASRec:Kang:2018}.
	}
	\label{fig-baselines}
\end{figure}

With the growth of the Internet, 
the amount of information being generated every moment 
is far beyond human discernment.
Recommender systems are thus developed to help humans quickly and accurately ascertain the items of interest,
and have played important roles in diverse applications,
including e-commerce~\cite{zhou2018deep}, online news~\cite{gong2022positive}, and education~\cite{WuWZ24}.
Due to inherent differences in data types and recommendation goals, 
different tasks are typically handled using various techniques and separate models.
For example, graph neural networks \cite{GNN:GCN:Kipf:2016} have dominated collaborative filtering \cite{CF:LightGCN:He:2020,CF:UltraGCN:Mao:2021},
while Transformer \cite{Transformer:Attention:Vaswani:2017} becomes increasingly popular in sequential recommendation \cite{Seq:SASRec:Kang:2018,Seq:Bert4Rec:Sun:2019}.

Recently, the prosperity of Large Language Models (LLMs) \cite{LLM:T5:Raffel:2020,LLM:GPT:Openai:2023,LLM:LLama:Touvron:2023,LLM:LLaMA2:Touvron:2023} 
suggests a promising direction towards universal recommenders~\cite{LLM4Rec:P5:Geng:2022,LLM4Rec:POD:Li2023}.
Equipped with carefully designed prompts, 
they show great potential in explainable and cross-domain recommendations \cite{LLM4Rec:CRs:Feng:2023,LLM4Rec:Chat-Rec:Gao:2023}.
Nevertheless, 
there still exist non-negligible gaps \cite{LLM4Rec:Prefrence:Kang:2023,LLM4Rec:TallRec:Bao:2023} between LLMs and conventional methods
unless domain-specific knowledge is injected.
Some researches have observed `compelling' results after fine-tuning~\cite{LLM4Rec:E4SRec:Li:2023,LLM4Rec:LLaMARec:Yue2023},
and hastily affirmed LLM-based recommenders' ranking capability.

However, the comparisons therein are not objective and fair enough, 
leading to \textit{under-estimation} of conventional recommenders and \textit{over-confidence} in LLMs.
Recall that the next-token prediction objective used for LLM pre-training (and fine-tuning), by its nature, 
is a cross-entropy loss that needs a full softmax over the entire corpus.
In view of the substantial quantity of items in reality,
conventional methods typically adopt a pointwise/pairwise loss function (\eg, BCE and BPR).
This compromise however causes significant performance degradation.
As shown in Figure \ref{fig-baselines}, SASRec trained with cross-entropy outperforms LLMs by a large margin,
while falling behind with BCE or BPR.
Such superior results relying on cross-entropy however cannot serve as direct evidence to challenge the ranking capability of existing LLM-based recommenders, 
since the full softmax is intractable to calculate in practice.

In this work, 
we re-emphasize the ability to optimize ranking metrics for a desired recommendation loss,
and then unveil the corresponding limitations of some approximations to cross-entropy.
To achieve effective and practical approximations, we introduce some novel alternatives with theoretical analysis.
In summary, the innovative insights and technical contributions are as follows:
\begin{itemize}[leftmargin=*]
\item
\textbf{
    Minimizing cross-entropy is equivalent to maximizing a lower bound of 
    Normalized Discounted Cumulative Gain (NDCG) and Reciprocal Rank (RR).
}
One can thus expect that the ranking capability would be gradually enhanced as cross-entropy is optimized during training.
We further show that dynamic truncation on the normalizing term yields a tighter bound and potentially better performance.
This fact highlights the importance of optimizing these ranking metrics,
and the cross-entropy loss is arguably adequate for this purpose.
The challenge that remains unsolved is how to realize the approximation in an effective and practical manner,
so the comparison with LLM-based recommenders is meaningful in reality.
After revisiting the limitations of some well-known approximations,
a rather simple solution will be presented.

\item
\textbf{
    Noise contrastive estimation (NCE) \cite{loss:NCE:Gutmann:2010} with the default setting
    fails to optimize a meaningful bound in the early stages of training.
}
Before the advent of subword segmentation algorithms \cite{Subword:Unigram:Kudo:2018}, 
the training of neural language models also struggles to circumvent an explicit normalizing over the entire vocabulary.
Mnih et al. \cite{loss:NCE4LM:Mnih:2012} thus resorted to a simplified NCE
that fixes the normalizing term estimate as a constant value of~1.
This suggestion however introduces training difficulties in recommendation:
sampling more negative samples should accelerate the training yet the opposite occurs.
This intriguing phenomenon is attributed to the weak connection between NCE and NDCG (RR).
Because NCE grows exponentially fast w.r.t. the number of positively scored items,
a meaningless bound is encountered in the early training stages.
This conclusion suggests adjusting the estimate of the normalizing term to a slightly larger value, 
which shows promising empirical performance 
but lacks consistent applicability.
Next, we introduce a more reliable loss.

\item
\textbf{
    Scaling up the sampled normalizing term provides an effective and practical approximation to cross-entropy.
}
Since the normalizing term of cross-entropy is intractable in reality,
a direct way is to approximate it by (uniformly) sampling part of items (a.k.a. sampled softmax loss \cite{loss:SSM:Wu:2023}).
To further mitigate the magnitude loss caused by sampling, 
we multiply it by an additional weight so the sampled term is scaled up.
Indeed, this modification can also be understood as a special case of importance sampling \cite{loss:IS:Bengio:2008},
in which the proposal distribution assigns a higher probability mass to the target item.
Unlike NCE, this Scaled Cross-Entropy (dubbed SCE) yields a bound mainly determined by the current rank of the target item,
making it meaningful even in the early training stages.
Empirically, sampling a very few negative samples per iteration is sufficient to achieve comparable results to using cross-entropy with a full softmax.
\end{itemize}

Based on these approximations for cross-entropy, 
we conduct a comprehensive investigation to assess the true ranking capability of both conventional and LLM-based recommenders.
The experimental results presented in Section 5 suggest the over-confidence in existing LLM-based methods.
Even without considering the model sizes, 
they are still far inferior to conventional methods in next-item recommendation.
Apart from the potential of explainability and cross-domain transferability,
further investigation and exploration are necessary 
to assess the true ranking capability of LLM-based recommenders.

\section{Related Work}

\textbf{Recommender systems} 
are developed to enable users to quickly and accurately ascertain relevant items.
The primary principle is to learn underlying interests from user information, especially historical interactions.
Collaborative filtering \cite{loss:BPR:Rendle:2012,CF:NCF:He:2017} performs personalized recommendation 
by mapping users and items into the same latent space in which interacted pairs are close.
Beyond static user representations,
sequential recommendation \cite{shani2005mdp,li2020time} focuses on capturing dynamic interests from item sequences.
Early efforts such as GRU4Rec \cite{Seq:GRU4Rec:Hidasi:2015} and Caser~\cite{Seq:Caser:Tang:2018} 
respectively apply recurrent neural networks (RNNs) and convolutional neural networks (CNNs) to sequence modeling.
Recently, Transformer \cite{Transformer:Attention:Vaswani:2017,devlin2018bert} becomes increasingly popular in recommendation 
due to its parallel efficiency and superior performance.
For example, SASRec~\cite{Seq:SASRec:Kang:2018} and BERT4Rec \cite{Seq:Bert4Rec:Sun:2019} 
respectively employ unidirectional and bidirectional self-attention.
Differently,
Zhou et al. \cite{Seq:FMLPRec:Zhou:2022} present FMLP-Rec to denoise the item sequences through learnable filters 
so that state-of-the-art performance can be obtained by mere MLP modules.

\textbf{LLM for recommendation} 
has gained a lot of attention recently because:
1) The next-token generation feature is technically easy to extend to the next-item recommendation (\ie, sequential recommendation);
2) The immense success of LLM in natural language processing promises the development of universal recommenders.
Some studies \cite{LLM4Rec:UncoverChatGPT:Dai:2023,LLM4Rec:Chat-Rec:Gao:2023}
have demonstrated the powerful zero/few-shot ability of LLMs (\eg, GPT \cite{LLM:GPT:Openai:2023}),
especially their potential in explainable and cross-domain recommendations \cite{LLM4Rec:CRs:Feng:2023,LLM4Rec:Chat-Rec:Gao:2023}.
Nevertheless, there is a consensus \cite{LLM4Rec:Prefrence:Kang:2023,LLM4Rec:InstructRec:Zhang:2023,LLM4Rec:TallRec:Bao:2023} 
that without domain-specific knowledge learned by fine-tuning, 
LLM-based recommenders still stay far behind conventional models.

As an early effort, 
P5~\cite{LLM4Rec:P5:Geng:2022} unifies multiple recommendation tasks into a sequence-to-sequence paradigm.
Based on the foundation model of T5~\cite{LLM:T5:Raffel:2020}, 
each task can be activated through some specific prompts.
Hua et al.~\cite{LLM4Rec:CIDIID:Hua2023} takes a further step beyond P5 
by examining the impact of various ID indexing methods,
and a combination of collaborative and independent indexing stands out.
Recently, more LLM recommenders \cite{LLM4Rec:ControlRec:Qiu:2023,LLM4Rec:LLaRA:Liao:2023,LLM4Rec:E4SRec:Li:2023,LLM4Rec:LLaMARec:Yue2023} 
based on Llama~\cite{LLM:LLama:Touvron:2023} or Llama2~\cite{LLM:LLaMA2:Touvron:2023} are developed.
For example, 
LlamaRec~\cite{LLM4Rec:LLaMARec:Yue2023} proposes a two-stage framework based on Llama2 to rerank the candidates retrieved by conventional models.
To enable LLM to correctly identify items,
E4SRec~\cite{LLM4Rec:E4SRec:Li:2023} incorporates ID embeddings trained by conventional sequential models through a linear adaptor,
and applies LORA~\cite{LLM:LORA:Hu:2021} for parameter-efficient fine-tuning.

\textbf{Cross-entropy and its approximations} \cite{loss:BPR:Rendle:2012,loss:NCE:Gutmann:2010,loss:IS:Bengio:2008,loss:SSM:Blanc:2018,loss:Consistency:Ma2018}
have been extensively studied.
The most related works are:
1) Bruch et al.~\cite{loss:Bound:Bruch:2019} theoretically connected cross-entropy to some ranking metrics;
2) Wu et al.~\cite{loss:SSM:Wu:2023} further found its desirable property in alleviating popularity bias;
and recently, 
3) Klenitskiy et al.~\cite{loss:NEG4Rec:Klenitskiy:2023} and Petrov et al.~\cite{loss:gBCE:Petrov:2023} 
respectively
applied cross-entropy and a generalized BCE loss to eliminate the performance gap between SASRec and BERT4Rec.
Differently, we are to
1) understand the superiority of cross-entropy as well as the limitations of its approximations;
2) identify a viable approximation according to these findings;
and 3) facilitate an objective evaluation of LLM-based recommendation by acknowledging the true capability of conventional models.

\section{Preliminaries}

Given a query $q$ that encompasses some user information,
a recommender system aims to retrieve some items $v \in \mathcal{I}$ that would be of interest to the user.
In sequential recommendation, 
the recommender predicts the next item $v_{t+1}$ based on historical interactions $q = [v_1, v_2, \ldots, v_t]$.
The crucial component is to develop a scoring function $s_{qv} := s_{\theta}(q, v)$ to accurately model the relevance of a query $q$ to a candidate item $v$.
A common paradigm is to map them into the same latent space through some models parameterized via $\theta$, 
followed by an inner product operation for similarity calculation.
Then, top-ranked items based on these scores will be prioritized for recommendation.
Typically the desired recommender is trained to minimize an objective function over all observed interactions $\mathcal{D}$: 
\begin{equation*}
    \min_{\theta} \quad \mathbb{E}_{(q, v_+) \sim \mathcal{D}} [\ell(q, v_+; \theta)],
\end{equation*}
where $v_+$ indicates the target item for the query $q$, 
and the loss function $\ell$ considered in this paper is in the form of
\begin{equation}
    \label{eq-loss-unified}
    \ell(q, v_+; \theta) := -\log \frac{
        \exp(s_{\theta}(q, v_+))
    }{
        Z_{\theta}(q)
    }
    = -s_{\theta}(q, v_+) + \log Z_{\theta}(q).
\end{equation}
Here $Z_{\theta}(q)$ depicts the `normalizing' term specified to the query $q$.
For clarity, we will omit $q$ and $\theta$ hereafter if no ambiguity is raised.

It is worth noting that
the reformulation of Eq. \eqref{eq-loss-unified} makes it easy to understand the subtle differences between a wide range of loss functions.
Table~\ref{table-summary-loss} covers a selection of approximations:
Binary Cross-Entropy (BCE) and Bayesian Personalized Ranking (BPR)~\cite{loss:BPR:Rendle:2012}
are widely used in recommendation for their low costs;
Importance Sampling (IS) \cite{loss:IS:Bengio:2008,loss:IS4LM:Jean:2014},
Noise Contrastive Estimation (NCE) \cite{loss:NCE:Gutmann:2010,loss:NCE4LM:Mnih:2012}, 
and NEGative sampling (NEG)~\cite{loss:Word2Vec:Mikolov:2013} 
are the cornerstones of the subsequent methods \cite{loss:IS4Rec:Lian:2020,loss:IS4Rec:Chen:2022,loss:Sampling:Gao:2021,loss:Sampling4Rec:Shi:2023,loss:SSM:Yi:2019,loss:SSM:Blanc:2018,loss:SSM:Yang:2020}.
Additional details regarding their mechanisms are provided in Appendix \ref{section-loss-function}.
Note that in this work we only involve some elementary approximations, as the primary purpose is not to develop a complex loss.
For more advanced loss functions, please refer to \cite{loss:NEG4Rec:Chen:2023}.

\begin{table*}[]
    \centering
    \caption{
    Cross-entropy loss and its approximations.
    The bounding probabilities are obtained in the case of uniform sampling.
    More conclusions in terms of Reciprocal Rank (RR) can be found in Appendix \ref{section-proofs}.
    }
    \label{table-summary-loss}
\scalebox{0.98}{
\begin{tabular}{l||c|c|c|c}
    \toprule
    Loss  
    & Formulation 
    & `Normalizing' term $Z$ 
    & Complexity 
    & $\mathbb{P}\bigl( -\log \text{NDCG}(r_+) \le \ell_* \bigr) \ge$ 
    \\
    \midrule
    $\ell_{\text{CE}}$
    & $ -\log \frac{\exp(s_{v_+})}{ \sum_{v \in \mathcal{I}} \exp(s_v)}$ 
    & $ \sum_{v \in \mathcal{I}} \exp(s_v)$
    & $\mathcal{O}(|\mathcal{I}|d)$
    & $1$
    \\
    \midrule
    $\ell_{\text{BCE}}$
    & $ -\log \sigma(s_{v_+}) - \log (1 - \sigma(s_{v_-}))$ 
    & $ (1 + \exp(s_{v_+})) (1 + \exp(s_{v_-}))$
    & $\mathcal{O}(d)$
    & -
    \\
    $\ell_{\text{BPR}}$
    & $ -\log \sigma(s_{v_+} - s_{v_-})$ 
    & $ \exp(s_{v_+}) + \exp(s_{v_-})$
    & $\mathcal{O}(d)$
    & -
    \\
    \midrule
    $\ell_{\text{NCE}}$
    & $ -\log \sigma(s_{v_+}') - \sum_{i=1}^K \log (1 - \sigma(s_{v_i}')) $ 
    & $ (1 + \exp(s_{v_+}')) \prod_{i=1}^K (1 + \exp(s_{v_i}'))$
    & $\mathcal{O}(Kd)$
    & $1 - m\big(1 - |\mathcal{S}_+'| / |\mathcal{I}| \big)^{\lfloor K/m \rfloor}$
    \\
    $\ell_{\text{NEG}}$
    & $ -\log \sigma(s_{v_+}) - \sum_{i=1}^K \log (1 - \sigma(s_{v_i})) $ 
    & $ (1 + \exp(s_{v_+})) \prod_{i=1}^K (1 + \exp(s_{v_i}))$
    & $\mathcal{O}(Kd)$
    & $1 - m\big(1 - |\mathcal{S}_+| / |\mathcal{I}| \big)^{\lfloor K/m \rfloor}$
    \\
    \midrule
    $\ell_{\text{IS}}$
    & $ -\log \frac{\exp(s_{v_+} - \log Q(v_+))}{\sum_{i=1}^K \exp(s_{v_i} - \log Q(v_i))} $ 
    & $ \sum_{i=1}^K \exp(s_{v_i} - \log Q(v_i) + \log Q(v_+))$
    & $\mathcal{O}(Kd)$
    & $1 - 2^m\big(1 - r_+ / |\mathcal{I}| \big)^{\lfloor K / 2^m \rfloor}$
    \\
    $\ell_{\text{SCE}}$
    & $ -\log \frac{\exp(s_{v_+})}{\exp(s_{v_+}) + \alpha \sum_{i=1}^K \exp(s_{v_i})} $ 
    & $ \exp(s_{v_+}) + \alpha \sum_{i=1}^K \exp(s_{v_i}) $
    & $\mathcal{O}(Kd)$
    & $1 -\frac{1}{\alpha}2^m\big(1 - r_+ / |\mathcal{I}| \big)^{\lfloor \alpha K / 2^m \rfloor}$
    \\
    \bottomrule
\end{tabular}
}
\end{table*}

\section{The Role of Cross-Entropy Loss in Optimizing Ranking Capability}
\label{section-methodology}

BCE and BPR are commonly employed in training recommender systems due to their high efficiency.
However, the substantial performance gaps shown in Figure~\ref{fig-baselines} suggest their poor alignment with cross-entropy.
Needless to say, claims based on the comparison with these inferior alternatives are unconvincing.
In this section, we are to showcase the substitutability of cross-entropy by
1) highlighting the importance of implicitly optimizing the ranking metrics for a recommendation loss;
2) introducing some practical modifications to boost the effectiveness of some elementary approximations.
Due to space constraints, the corresponding proofs are deferred to Appendix \ref{section-proofs}.

SASRec \cite{Seq:SASRec:Kang:2018}, one of the most prominent sequential models, 
will serve as the baseline to empirically elucidate the conclusions in this part.
All results are summarized based on 5 independent runs.

\subsection{Cross-Entropy for Some Ranking Metrics}

\begin{figure}
	\centering
	\includegraphics[width=0.47\textwidth]{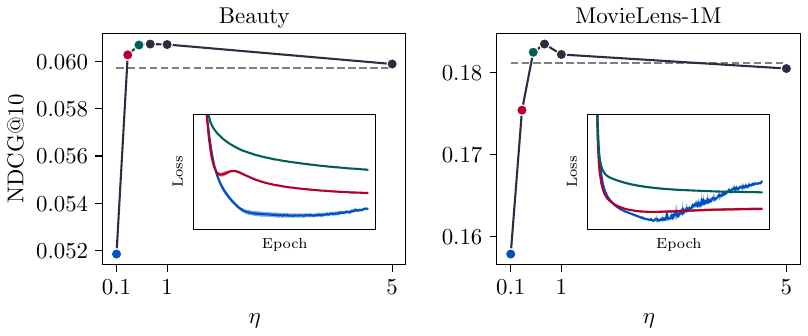}
	\caption{
        Performance comparison based on tighter bounds for NDCG.
        The dashed line represents the results trained by CE (namely the case of $\eta \rightarrow +\infty$).
	}
	\label{fig-tighter-bounds}
\end{figure}

The capability to prioritize items aligning with the user's interests is essential for recommender systems.
Denoted by $r_+ := r(v_+) = |\{v \in \mathcal{I}: s_v \ge s_{v_+}\}|$ the predicted rank of the target item $v_+$,
the metric of Normalized Discounted Cumulative Gain (NDCG)\footnote{
In practice, it is deemed meaningless when $r_+$ exceeds a pre-specified threshold $k$ (\eg, $k=1,5,10$).
Hence, the widely adopted NDCG@$k$ metric is modified to assign zero reward to these poor ranking results.
}
is often employed to assess the sorting quality.
For next-item recommendation considered in this paper, NDCG is simplified to
\begin{equation*}
    \text{NDCG}(r_+) = \frac{1}{\log_2 (1 + r_+)}.
\end{equation*}
It increases as the target item $v_+$ is ranked higher, and reaches the maximum when $v_+$ is ranked first (\ie, $r_+ = 1$).
Consequently, 
the average quality computed over the entire test set serves as an indicator of the ranking capability.
Notably, Reciprocal Rank (RR) is another popular ranking metric,
and we leave the definition and results in Appendix 
since the corresponding findings are very similar to those of NDCG.
The following proposition suggests that cross-entropy is a soft proxy to these ranking metrics.
\begin{proposition}
    \label{proposition-bounding}
    For a target item $v_+$ which is ranked as $r_+$, the following inequality holds true for any $n \ge r_+$
    \begin{equation}
        -\log \mathrm{NDCG}(r_+) \le \ell_{\mathrm{CE}\text{-}n},
    \end{equation}
    where
    \begin{equation*}
        \ell_{\mathrm{CE}\text{-}n} := s_{v_+}  + \log \sum_{r(v) \le n} \exp(s_{v}).
    \end{equation*}
\end{proposition}

We can draw from Proposition \ref{proposition-bounding} that $-\log \mathrm{NDCG}(r_+)$ would be strictly bounded by CE-like losses, 
as long as all items ranked before $v_+$ are retained in the normalizing term.
In other words, minimizing these CE-like losses is equivalent to maximizing a lower bound of NDCG.
Because cross-entropy is a special case that retains all items (\ie, $n=|\mathcal{I}|$), 
we readily have the following corollary:

\begin{corollary}[\cite{loss:Bound:Bruch:2019}]
    Minimizing the cross-entropy loss $\ell_{\mathrm{CE}}$ is equivalent to maximizing a lower bound of normalized discounted cumulative gain.
\end{corollary}

Therefore, satisfactory ranking capability can be expected if $\ell_{\text{CE}}$ for all queries are minimized.
Since the superiority of cross-entropy possibly stems from its connection to some ranking metrics, 
one may hypothesize that optimizing a tighter bound with a smaller value of $n \ll |\mathcal{I}|$ allows greater performance gains.
However, the condition $n \ge r_+$ cannot be consistently satisfied via a constant value of $n$ since $r_+$ dynamically changes during training.
Alternatively, an adaptive truncation can be employed for this purpose:
\begin{equation}
    \label{eq-truncation}
    \ell_{\text{CE-}\eta} := s_{v_+} + \log \sum_{s_v - s_+ \ge -\eta |s_+|} \exp(s_{v}), \quad \eta \ge 0.
\end{equation}
Note that this \textit{$\eta$-truncated} loss retains only items whose scores are not lower than $s_+ - \eta |s_+|$,
so a tighter bound will be obtained as $\eta$ drops to 0.
Specifically, this $\eta$-truncated loss becomes $\ell_{\text{CE-}r_+}$ (\ie, the tightest case) when $\eta = 0$, and approaches $\ell_{\text{CE}}$ when $\eta \rightarrow +\infty$.
Figure \ref{fig-tighter-bounds} illustrates how NDCG@10 varies as $\eta$ gradually increases from 0.1 to 5.
There are two key observations:
\begin{itemize}[leftmargin=*]
    \item[1.]
    SASRec performs worst in the tightest case of $\eta \approx 0$.
    This can be attributed to the instability of the $\eta$-truncated loss.
    On the one hand,
    $\ell_{\text{CE-}\eta}$ will rapidly collapse to 0 for those easily recognized targets, 
    in which case all other items are excluded in the normalizing term except the target itself.
    On the other hand,
    due to this strict truncation, 
    only a few negative items are encountered during training,
    and thus over-fitting is more likely to occur.

    \item[2.]
    Once $\eta$ is large enough to overcome the training instability, 
    SASRec begins to enjoy the benefits from tightness and achieves its best performance around $\eta \approx 0.7$.
    Further increasing $\eta$ however leads to a similar effect as cross-entropy, thereby along with a slightly degenerate performance due to the suboptimal tightness.
\end{itemize}

In spite of the minor performance gains,
the complexity of these tighter bounds is still equal to or even higher than that of cross entropy.
The challenge that remains unsolved is how to realize the approximation in an effective and practical manner.
To this end, we will introduce two practical alternatives to cross-entropy,
one based on noise contrastive estimation \cite{loss:NCE:Gutmann:2010,loss:NCE4LM:Mnih:2012}, 
and the other based on sampled softmax loss \cite{loss:SSM:Wu:2023}.
Their different ways of approximating the cross-entropy loss lead to distinct properties during training.
Some specific modifications focusing on the normalizing term estimates are then developed to enhance their ability to optimize NDCG and RR.

\subsection{Revisiting Noise Contrastive Estimation}
\label{section-nce}

\begin{figure*}
	\centering
    \subfloat[Beauty]{
        \includegraphics[width=0.47\textwidth]{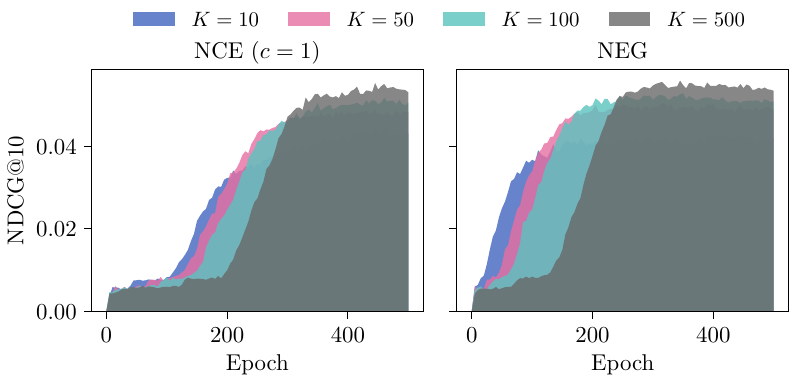}
    }
    \subfloat[MovieLens-1M]{
        \includegraphics[width=0.47\textwidth]{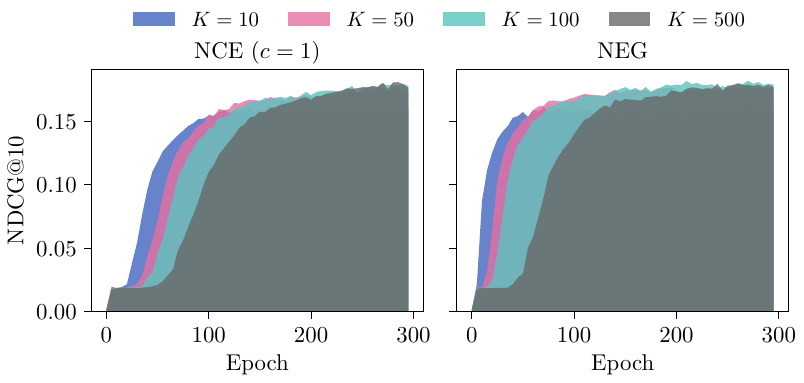}
    }
	\caption{
        NDCG@10 performance of NCE and NEG across different number of negative samples.
	}
	\label{fig-nce-neg-negs-ndcg}
\end{figure*}

Noise Contrastive Estimation (NCE) is widely used in training neural language models 
for bypassing an explicit normalizing over the entire vocabulary.
It requires the model to discriminate the target from an easy-to-sample noise distribution. 
In the case of the uniform sampling, it can be formulated as follows
\begin{equation*}
    \ell_{\text{NCE}} := 
    s_{v_+} + \log \bigg\{
        \big(1 + \exp(s_{v_+}') \big) 
        \prod_{i=1}^K \big(1 + \exp(s_{v_i}')\big)
    \bigg\},
\end{equation*}
where $s_{v}' = s_{v} - c - \log \frac{K}{|\mathcal{I}|}$.
In the original implementation of NCE~\cite{loss:NCE:Gutmann:2010}, $c$ is a trainable parameter as an estimate of $\log Z_{\text{CE}}$.
However, 
this strategy is infeasible for conditional probability models like language models and sequential recommendation, 
where they need to determine one $c_q$ for each text (query).
Mnih et al. \cite{loss:NCE4LM:Mnih:2012} therefore fixed $c \equiv 1$ for all texts during training,
and NEGative sampling (NEG) used in Word2Vec~\cite{loss:Word2Vec:Mikolov:2013} further simplifies it by replacing $s'$ with $s$ directly; that is,
\begin{equation*}
    \ell_{\text{NEG}} := 
    s_{v_+} + \log \bigg\{
        \big(1 + \exp(s_{v_+}) \big) 
        \prod_{i=1}^K \big(1 + \exp(s_{v_i})\big)
    \bigg\}.
\end{equation*}

However, we observe that both NCE ($c=1$) and NEG introduce training difficulties as the number of negative samples increases.
The NDCG@10 metric shown in Figure \ref{fig-nce-neg-negs-ndcg} remains almost constant at the beginning of training,
and consumes more iterations to converge as $K$ increases.
This contradicts the usual understanding that
sampling more negative samples would accelerate the convergence of training.
We believe that this intriguing phenomenon stems mainly from the exponential growth of the normalizing term,
which consequently yields a rather weak bound:

\begin{theorem}
    \label{theorem-nce-neg-bounds}
    Let $v_+$ be a target item which is ranked as $r_+ \le 2^{2^m} - 1$ for some $m \in \mathbb{N}$,
    and
    \begin{align*}
        \mathcal{S}_+ := \{v \in \mathcal{I}: s_{v} \ge 0\},
        \quad 
        \mathcal{S}_+' := \{v \in \mathcal{I}: s_{v}' \ge 0\}.
    \end{align*}
    If we uniformly sampling $K$ items for training, then with probability at least
    \begin{equation}
        \label{eq-nce-neg-probs}
        \left \{
        \begin{array}{ll}
            1 - m\big(1 - |\mathcal{S}_+'| / |\mathcal{I}| \big)^{\lfloor K/m \rfloor}, & \text{ if } \ell_* = \ell_{\mathrm{NCE}} \\
            1 - m\big(1 - |\mathcal{S}_+| / |\mathcal{I}| \big)^{\lfloor K/m \rfloor}, & \text{ if } \ell_* = \ell_{\mathrm{NEG}} \\
        \end{array}
        \right .,
    \end{equation}
    we have
    \begin{equation}
        \label{eq-nce-neg-bounds}
        -\log \mathrm{NDCG}(r_+) \le \ell_*.
    \end{equation}
\end{theorem}

From Theorem \ref{theorem-nce-neg-bounds}, we have the following conclusions:
\begin{itemize}[leftmargin=*]
\item[1.]
Notice that
Eq. \eqref{eq-nce-neg-probs} for NCE (NEG) is mainly determined by the size of $|S_+'|$ ($|S_+|$) rather than the current rank $r_+$.
As a result, NCE and NEG can easily bound NDCG
as long as the number of items with non-negative scores is large enough.
This is more common in the early stages of training, 
in which case item representations are poorly distributed in the latent space.
In other words, 
the bounds at the beginning are too weak to be meaningful for improving the model ranking capability.
The so-called training difficulties are actually a stage in narrowing the gap.
\item[2.]
Figure \ref{fig-nce-neg-negs-ndcg} also suggests that the standstill duration of NCE 
is significantly longer than that of NEG,
for example, 150 epochs versus 70 epochs on Beauty if 500 negative samples are sampled for training.
Note that NEG can be regarded as a special case of NCE by fixing $c = \log (|\mathcal{I}|/K)$,
a value higher than 1 if the experimental settings described in Figure~\ref{fig-nce-neg-negs-ndcg} are applied.
As such, according to Theorem \ref{theorem-nce-neg-bounds},
NCE with $c=1$ will suffer from a weaker bound than NEG, thereby more iterations are required for convergence.
\end{itemize}

\begin{table}
    \centering
    \caption{
        The convergence epoch and the highest NDCG@10 metric achieved among the 500 epochs (with $K=500$).
    }
    \label{table-nce-c}
\begin{tabular}{l|cccc}
    \toprule
    & \multicolumn{2}{c}{Beauty}                          & \multicolumn{2}{c}{MovieLens-1M}                    \\
    \midrule
    & NDCG@10          & Epoch          & NDCG@10   & Epoch     \\
    \midrule
    NCE ($c=1$) & 0.0547 & 400-470 & 0.1817 & 280-440 \\
    NCE ($c=5$) & 0.0545 & 195-280 & 0.1812 & 100-200 \\
    NCE ($c=10$) & \textbf{0.0571} & \textbf{95-115} & \textbf{0.1817} & \textbf{80-160}  \\
    NCE ($c=50$) & 0.0410 & $\ge 500$ & 0.1816 & $\ge 420$ \\
    NCE ($c=100$) & 0.0171 & $\ge 500$ & 0.1676 & $\ge 500$ \\
    \bottomrule
\end{tabular}
\end{table}

Overall, NCE and NEG converge slower as $K$ increases because of the exponential growth of their normalizing terms w.r.t. the sizes of $|\mathcal{S}_+'|$ and $|\mathcal{S}_+|$.
One feasible modification is to adopt a moderately large $c$ 
so that the sizes remain acceptable even if numerous negative items are sampled.
As shown in Table \ref{table-nce-c}, 
the number of epochs required for convergence decreases as the value of $c$ increases from 1 to 10.
But a larger value once again hinders the training process.
Recall that $c$ is an estimate of $\log Z_{\text{CE}}$.
Setting $c \ge 50$ implies a hypothesis of $Z_{\text{CE}} \ge e^{50}$, 
which is obviously difficult to achieve for most models.
As a rule of thumb,
the hyper-parameter $c$ should be chosen carefully,
and $c=10$ appears a good choice according to the experiments in Section \ref{section-experiments}.
Next we will introduce a more reliable variant of the sampled softmax loss \cite{loss:SSM:Wu:2023}.
As opposed to NCE and NEG,
it yields a bound determined by the current rank $r_+$.

\subsection{Scaling Up the Sampled Normalizing Term}

Since the normalizing term of cross-entropy is intractable in reality,
a direct way is to approximate it by (uniformly) sampling part of items from $\mathcal{I}$:
\begin{equation}
    \hat{Z}(\alpha) = \exp(s_{v_+}) + \alpha \sum_{i=1}^K \exp(s_{v_i}), \quad \alpha \ge 1.
\end{equation}
The resulting Scaled Cross-Entropy (SCE) becomes
\begin{equation*}
    \ell_{\text{SCE}} := -s_{v_+} + \log \hat{Z}(\alpha).
\end{equation*}
Note that here we scale up the sampled normalizing term by a pre-specific weight $\alpha$.
For the case of $\alpha=1$,
this approximation (a.k.a. sampled softmax loss~\cite{loss:SSM:Wu:2023}) implies a $(K+1)$-class classification task,
and has been used in previous studies \cite{loss:NEG4Rec:Klenitskiy:2023,loss:SSM:Wu:2023}.
But it is worth noting that they have inherent differences.
We modify it via a weight $\alpha$ to remedy the magnitude loss resulted from sampling,
so the scaled loss is more likely to bound NDCG and RR:

\begin{theorem}
    \label{theorem-sce-bounds}
    Under the same conditions as stated in Theorem \ref{theorem-nce-neg-bounds},
    the inequality \eqref{eq-nce-neg-bounds} holds for SCE with a probability of at least
    \begin{equation}
        \label{eq-sce-bounds}
        1 -\frac{1}{\alpha}2^m\big(1 - r_+ / |\mathcal{I}| \big)^{\lfloor \alpha K / 2^m \rfloor}.
    \end{equation}
\end{theorem}

In addition to the promising bounding probability achieved through the weight $\alpha$,
we can also observe from Theorem \ref{theorem-sce-bounds} that
Eq.~\eqref{eq-sce-bounds} is directly governed by the current rank $r_+$.
In contrast to NCE ($c=1$) and NEG,
SCE is expected to be meaningful even in the early stages of training.
Certainly, SCE is not perfect:
the scaling operation inevitably raises concerns about the high variance problem.
As depicted in Figure~\ref{fig-sce-alpha-negs-ndcg}, 
if negative samples are very rare, 
a larger weight of $\alpha$ tends to worsen the ranking capability.
Fortunately,
the high variance problem appears less significant 
as $K$ slightly increases (\eg, $K \ge 50$ for Beauty and $K \ge 10$ for MovieLens-1M).
Notably,
sampling 100 negative samples for $\alpha = 100$ produces comparable performance to using 500 negative samples for $\alpha = 1$ on the Beauty dataset.

\begin{figure}[t]
	\centering
	\includegraphics[width=0.47\textwidth]{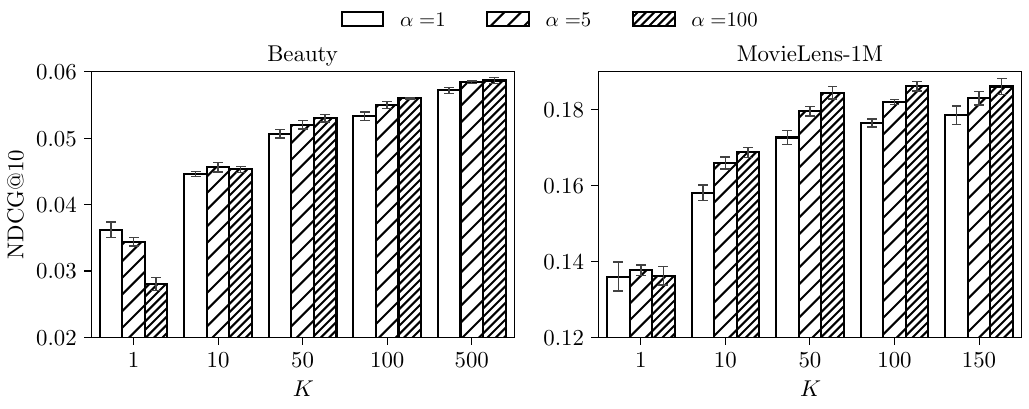}
	\caption{
        NDCG@10 performance under various weight $\alpha$.
	}
	\label{fig-sce-alpha-negs-ndcg}
\end{figure}

\textbf{Connection to importance sampling.}
While SCE does not make sense at first glance, 
it is indeed closely related to importance sampling \cite{loss:IS:Bengio:2008,robert1999monte},
a widely used technique for cross-entropy approximation.
Given a proposal distribution $Q$ over all items, it corrects the approximation as follows
\begin{equation*}
    \ell_{\text{IS}} := -\log \frac{
        \exp(s_{v_+} - \log Q(v_+))
    }{
        \sum_{i=1}^K \exp(s_{v_i} - \log Q(v_i))
    }, \quad v_i \sim Q, \: i=1,2,\ldots, K.
\end{equation*}
It allows for an unbiased estimation if the proposal distribution $Q$ is precisely identical to the underlying data distribution.
But for the conditional probability model like sequential recommendation, achieving the optimum requires additional overhead for each query.
Hence, some heuristic designs based on popularity sampling \cite{loss:IS4Rec:Lian:2020,loss:IS4Rec:Chen:2022} are adopted more often in practice.
We point out that SCE is \textit{morphologically} equivalent to these designs with
\begin{equation}
    \label{eq-scaled-IS}
    Q(v) = 
    \left \{
    \begin{array}{ll}
        \frac{
            \alpha
        }{
            |\mathcal{I}| - 1 + \alpha
        } & \text{ if } v = v_+ \\
        \frac{
            1
        }{
            |\mathcal{I}| - 1 + \alpha
        } & \text{ if } v \not= v_+
    \end{array}
    \right..
\end{equation}
Consequently scaling up the sampled normalizing term can be considered as assigning a higher probability mass to the target item~$v_+$. 
This partially explains why SCE is effective:
a well-trained scoring function should skew towards the target item.
Another subtle difference is that the normalizing term for importance sampling may not include the target term, 
while SCE always preserves it.
This is practically necessary to avoid an unstable training process.

\subsection{Computational Complexity}

The major cost of cross-entropy lies in the inner product and softmax operations.
It has a complexity of $\mathcal{O}(|\mathcal{I}|d)$, where $d$ denotes the embedding size before similarity calculation.
In contrast, the approximations require a lower cost $\mathcal{O}(Kd)$, 
correspondingly an additional overhead $\mathcal{O}(K)$ for uniform sampling.
Overall, it is profitable if $K \ll |\mathcal{I}|$.

\section{Experiments}
\label{section-experiments}

In this section, 
we are to reveal the true ranking capability of conventional recommenders 
by using the modified Noise Contrastive Estimation (NCE) and the proposed Scaled Cross-Entropy (SCE).
Thus, the current advancements made by LLM-based recommenders can also be assessed objectively.

\begin{table}
  \centering
  \caption{Dataset statistics.}
  \label{table-dataset-statistics}
  \scalebox{0.8}{
\begin{tabular}{c|ccccc}
  \toprule
Dataset     & \#Users & \#Items & \#Interactions & Density & Avg. Length \\
  \midrule
Beauty      & 22,363  & 12,101  & 198,502        & 0.07\%  & 8.9       \\
MovieLens-1M & 6,040   & 3,416   & 999,611        & 4.84\%  & 165.5    \\
Yelp       & 30,431  & 20,033  & 316,354        & 0.05\%  & 10.4       \\
  \bottomrule
\end{tabular}
  }
\end{table}

\begin{table}[]
    \centering
    \caption{
        Model statistics.
        The number of parameters is estimated based on the Beauty dataset.
    }
    \label{table-model-statistics}
    \scalebox{0.75}{
\begin{tabular}{lcccc}
    \toprule
\multicolumn{1}{c}{Model} & Foundation Model & Architecture & Embedding Size & \#Params \\
    \midrule
P5(CID+IID) \cite{LLM4Rec:CIDIID:Hua2023}                              & T5                                & Transformer                   & 512                                & 60M                        \\
POD \cite{LLM4Rec:POD:Li2023}                                        & T5                                & Transformer                   & 512                                & 60M                        \\
LlamaRec \cite{LLM4Rec:LLaMARec:Yue2023}                                   & Llama2                            & Transformer                   & 4096                                & 7B                        \\
E4SRec \cite{LLM4Rec:E4SRec:Li:2023}                                     & Llama2                            & Transformer                   & 4096                                & 7B                    \\
    \midrule
GRU4Rec \cite{Seq:GRU4Rec:Hidasi:2015}                                    & -                                 & RNN                           & 64                              & 0.80M                   \\
Caser \cite{Seq:Caser:Tang:2018}                                      & -                                 & CNN                           & 64                              & 3.80M                   \\
SASRec \cite{Seq:SASRec:Kang:2018}                                     & -                                 & Transformer                   & 64                              & 0.83M                   \\
BERT4Rec \cite{Seq:Bert4Rec:Sun:2019}                                   & -                                 & Transformer                   & 64                              & 1.76M                   \\
FMLP-Rec \cite{Seq:FMLPRec:Zhou:2022}                                   & -                                 & MLP                           & 64                              & 0.92M                   \\
    \bottomrule
\end{tabular}
    }
\end{table}

\subsection{Experimental Setup}
\label{section-setup}

This part introduces the datasets, evaluation metrics, baselines, and implementation details.

\begin{table*}[]
    \caption{
        Overall performance comparison on the Beauty, MovieLens-1M, and Yelp datasets.
        The best results of each block are marked in bold.
        `\improv over CE/LLM' represents the relative gap between respective best results.
    }
    \label{table-overall-comparison}
    \centering
    \scalebox{0.83}{
\begin{tabular}{c|l|cccc|cccc|cccc}
    \toprule
\multicolumn{1}{l}{}    &            & \multicolumn{4}{c}{Beauty}                                            & \multicolumn{4}{c}{MovieLens-1M}                                      & \multicolumn{4}{c}{Yelp}                                              \\
    \midrule
\multicolumn{1}{l}{}    &            & HR@5            & HR@10           & NDCG@5          & NDCG@10         & HR@5            & HR@10           & NDCG@5          & NDCG@10         & HR@5            & HR@10           & NDCG@5          & NDCG@10         \\
    \midrule
    \midrule
\multirow{4}{*}{LLM}    & POD        & 0.0185          & 0.0245          & 0.0125          & 0.0146          & 0.0422          & 0.0528          & 0.0291          & 0.0326          & \textbf{0.0476}          & 0.0564          & \textbf{0.0330}          & 0.0358          \\
                        & P5 (CID+IID)    & 0.0569          & 0.0791          & 0.0403          & 0.0474          & \textbf{0.2225} & \textbf{0.3131} & \textbf{0.1570} & \textbf{0.1861} & 0.0289          & 0.0453          & 0.0200          & 0.0252          \\
                        & LlamaRec   & \textbf{0.0591} & \textbf{0.0862} & \textbf{0.0405} & \textbf{0.0492} & 0.1757          & 0.2836          & 0.1113          & 0.1461          & 0.0416 & \textbf{0.0605} & 0.0306 & \textbf{0.0367} \\
                        & E4SRec     & 0.0527          & 0.0753          & 0.0376          & 0.0448          & 0.1871          & 0.2765          & 0.1234          & 0.1522          & 0.0309          & 0.0473          & 0.0207          & 0.0260          \\
    \midrule
\multirow{5}{*}{CE}     & GRU4Rec    & 0.0474          & 0.0690          & 0.0329          & 0.0398          & 0.2247          & 0.3201          & 0.1542          & 0.1850          & 0.0275          & 0.0463          & 0.0171          & 0.0231          \\
                        & Caser      & 0.0435          & 0.0614          & 0.0303          & 0.0361          & 0.2181          & 0.3049          & 0.1520          & 0.1800          & 0.0283          & 0.0383          & 0.0211          & 0.0243          \\
                        & SASRec     & 0.0713          & 0.0986          & \textbf{0.0510} & \textbf{0.0597} & 0.2221          & 0.3131          & 0.1518          & 0.1812          & 0.0476          & 0.0696          & 0.0345          & 0.0415          \\
                        & BERT4Rec   & 0.0509          & 0.0747          & 0.0347          & 0.0423          & 0.1978          & 0.2922          & 0.1330          & 0.1634          & 0.0355          & 0.0540          & 0.0243          & 0.0303          \\
                        & FMLP-Rec   & \textbf{0.0717} & \textbf{0.0988} & 0.0507          & 0.0594          & \textbf{0.2287} & \textbf{0.3243} & \textbf{0.1585} & \textbf{0.1893} & \textbf{0.0512} & \textbf{0.0759} & \textbf{0.0364} & \textbf{0.0444} \\
\multicolumn{2}{c|}{\improv over LLM} & 21.4\%          & 14.6\%          & 25.7\%          & 21.3\%          & 2.8\%           & 3.6\%           & 0.9\%           & 1.7\%           & 7.5\%           & 25.6\%          & 10.3\%          & 21.0\%          \\
    \midrule
    \midrule
\multirow{5}{*}{BCE}    & GRU4Rec    & 0.0214          & 0.0376          & 0.0134          & 0.0186          & 0.1595          & 0.2490          & 0.1023          & 0.1310          & 0.0157          & 0.0273          & 0.0098          & 0.0135          \\
                        & Caser      & 0.0282          & 0.0434          & 0.0185          & 0.0234          & 0.1639          & 0.2476          & 0.1078          & 0.1348          & 0.0304          & 0.0428          & 0.0224          & 0.0264          \\
                        & SASRec     & 0.0429          & 0.0671          & 0.0275          & 0.0353          & 0.1594          & 0.2492          & 0.1040          & 0.1329          & 0.0325          & 0.0501          & 0.0225          & 0.0281          \\
                        & BERT4Rec   & 0.0245          & 0.0415          & 0.0152          & 0.0207          & 0.1241          & 0.2021          & 0.0789          & 0.1039          & 0.0223          & 0.0379          & 0.0138          & 0.0188          \\
                        & FMLP-Rec   & \textbf{0.0460} & \textbf{0.0710} & \textbf{0.0301} & \textbf{0.0381} & \textbf{0.1800} & \textbf{0.2722} & \textbf{0.1173} & \textbf{0.1469} & \textbf{0.0460} & \textbf{0.0651} & \textbf{0.0330} & \textbf{0.0391} \\
\multicolumn{2}{c|}{\improv over CE}  & -35.9\%         & -28.1\%         & -41.0\%         & -36.2\%         & -21.3\%         & -16.1\%         & -26.0\%         & -22.4\%         & -10.0\%         & -14.3\%         & -9.4\%          & -11.9\%         \\
\multicolumn{2}{c|}{\improv over LLM} & -22.1\%         & -17.7\%         & -25.8\%         & -22.6\%         & -19.1\%         & -13.1\%         & -25.3\%         & -21.1\%         & -3.3\%          & 7.6\%           & -0.1\%          & 6.5\%           \\
    \midrule
\multirow{5}{*}{NCE}    & GRU4Rec    & 0.0434          & 0.0652          & 0.0288          & 0.0359          & 0.2273          & 0.3184          & 0.1541          & 0.1834          & 0.0241          & 0.0418          & 0.0148          & 0.0205          \\
                        & Caser      & 0.0377          & 0.0567          & 0.0253          & 0.0314          & 0.2213          & 0.3106          & 0.1523          & 0.1811          & 0.0296          & 0.0405          & 0.0220          & 0.0255          \\
                        & SASRec     & 0.0686          & 0.0961          & 0.0485          & 0.0573          & 0.2177          & 0.3135          & 0.1479          & 0.1788          & 0.0471          & 0.0682          & 0.0344          & 0.0412          \\
                        & BERT4Rec   & 0.0487          & 0.0734          & 0.0324          & 0.0404          & 0.1960          & 0.2933          & 0.1311          & 0.1624          & 0.0389          & 0.0574          & 0.0271          & 0.0330          \\
                        & FMLP-Rec   & \textbf{0.0693} & \textbf{0.0964} & \textbf{0.0491} & \textbf{0.0578} & \textbf{0.2291} & \textbf{0.3279} & \textbf{0.1567} & \textbf{0.1885} & \textbf{0.0512} & \textbf{0.0760} & \textbf{0.0364} & \textbf{0.0444} \\
\multicolumn{2}{c|}{\improv over CE}  & -3.4\%          & -2.4\%          & -3.6\%          & -3.2\%          & 0.2\%           & 1.1\%           & -1.1\%          & -0.4\%          & 0.0\%           & 0.1\%           & 0.0\%           & 0.1\%           \\
\multicolumn{2}{c|}{\improv over LLM} & 17.3\%          & 11.8\%          & 21.3\%          & 17.4\%          & 2.9\%           & 4.7\%           & -0.2\%          & 1.3\%           & 7.5\%           & 25.7\%          & 10.3\%          & 21.1\%          \\
    \midrule
\multirow{5}{*}{SCE}    & GRU4Rec    & 0.0489          & 0.0694          & 0.0344          & 0.0410          & 0.2309          & 0.3248          & 0.1587          & 0.1891          & 0.0290          & 0.0487          & 0.0183          & 0.0246          \\
                        & Caser      & 0.0456          & 0.0628          & 0.0322          & 0.0377          & 0.2274          & 0.3135          & 0.1586          & 0.1864          & 0.0293          & 0.0404          & 0.0218          & 0.0253          \\
                        & SASRec     & 0.0698          & 0.0968          & 0.0500          & 0.0587          & 0.2273          & 0.3186          & 0.1567          & 0.1862          & 0.0472          & 0.0693          & 0.0339          & 0.0410          \\
                        & BERT4Rec   & 0.0540          & 0.0776          & 0.0372          & 0.0449          & 0.2078          & 0.3014          & 0.1405          & 0.1707          & 0.0414          & 0.0612          & 0.0283          & 0.0346          \\
                        & FMLP-Rec   & \textbf{0.0703} & \textbf{0.0979} & \textbf{0.0502} & \textbf{0.0591} & \textbf{0.2372} & \textbf{0.3284} & \textbf{0.1648} & \textbf{0.1942} & \textbf{0.0517} & \textbf{0.0779} & \textbf{0.0357} & \textbf{0.0441} \\
\multicolumn{2}{c|}{\improv over CE}  & -2.0\%          & -0.9\%          & -1.5\%          & -1.1\%          & 3.7\%           & 1.3\%           & 4.0\%           & 2.6\%           & 1.0\%           & 2.6\%           & -2.0\%          & -0.6\%          \\
\multicolumn{2}{c|}{\improv over LLM} & 19.0\%          & 13.6\%          & 23.8\%          & 20.0\%          & 6.6\%           & 4.9\%           & 5.0\%           & 4.4\%           & 8.6\%           & 28.8\%          & 8.1\%           & 20.2\%          \\
    \bottomrule
\end{tabular}
    }
\end{table*}

\textbf{Datasets.}
To ensure the reliability of the conclusions, 
we select three public datasets from different scenarios, including the Beauty, MovieLens-1M, and Yelp datasets.
Beauty is an e-commerce dataset extracted from Amazon reviews known for high sparsity;
MovieLens-1M is a movie dataset with a much longer sequence length;
Yelp collects abundant meta-data suitable for multi-task training.
Following \cite{Seq:FMLPRec:Zhou:2022,LLM4Rec:P5:Geng:2022},
we filter out users and items with less than 5 interactions,
and the validation set and test set are split in a \textit{leave-one-out} fashion, 
namely the last interaction for testing and the penultimate one for validation.
The dataset statistics are presented in Table~\ref{table-dataset-statistics}.

\textbf{Evaluation metrics.}
For each user, the scores returned by the recommender will be sorted in descending order to generate candidate lists. 
In addition to the aforementioned NDCG@$k$ (Normalized Discounted Cumulative Gain),
HR@$k$ (Hit Rate) that quantifies the proportion of successful hits among the top-$k$ recommended candidates
will also be included in this paper due to its widespread use in other studies \cite{Seq:SASRec:Kang:2018,Seq:FMLPRec:Zhou:2022,LLM4Rec:P5:Geng:2022}.
Besides, we also provide the Mean Reciprocal Rank (MRR) results in Appendix \ref{section-mrr}.

\textbf{Baselines.}
Although LLM itself has surprising zero-shot recommendation ability,
there still exist non-negligible gaps \cite{LLM4Rec:Prefrence:Kang:2023,LLM4Rec:TallRec:Bao:2023} unless domain-specific knowledge is injected.
Hence, only LLM-based recommenders enhanced by fine-tuning will be compared in this paper:
\textbf{P5 (CID+IID)}~\cite{LLM4Rec:CIDIID:Hua2023}, \textbf{POD}~\cite{LLM4Rec:POD:Li2023},
\textbf{LlamaRec}~\cite{LLM4Rec:LLaMARec:Yue2023}, and \textbf{E4SRec}~\cite{LLM4Rec:E4SRec:Li:2023}.
Specifically, the first two methods take T5 as the foundation model,
while the last two methods fine-tune Llama2 for efficient sequential recommendation.
Additionally, 
five sequential models including \textbf{FMLP-Rec} \cite{Seq:FMLPRec:Zhou:2022}, 
\textbf{Caser} \cite{Seq:Caser:Tang:2018}, 
\textbf{GRU4Rec} \cite{Seq:GRU4Rec:Hidasi:2015}, 
\textbf{SASRec}~\cite{Seq:SASRec:Kang:2018}, and \textbf{BERT4Rec} \cite{Seq:Bert4Rec:Sun:2019},
are considered here to unveil the true capability of conventional methods.
They cover various architectures so as to comprehensively validate the effectiveness of the proposed approximation methods.
Table~\ref{table-model-statistics} presents an overview of the model statistics.
Notice that for LlamaRec and E4SRec we employ Llama2-7B instead of Llama2-13B 
as the foundation model in order to improve training efficiency.
This results in minor performance differences in practice.

\textbf{Implementation details.}
According to the discussion in Section~\ref{section-methodology},
we set $c=10$ for NCE and $\alpha=100$ for SCE,
and less than 5\% of all items will be sampled for both approximation objectives.
Specifically, $K=500$ on the Beauty dataset, and $K=100$ on the MovieLens-1M and Yelp datasets.
We find that training 200 epochs is sufficient for cross-entropy to converge,
while sometimes NCE and SCE need 300 epochs.
Other hyper-parameters of NCE and SCE completely follow cross-entropy.
Because the data preprocessing scripts provided with POD may lead to information leakage \cite{LLM4Rec:Tiger:Rajput:2023},
we assign random integer IDs to items rather than sequentially incrementing integer IDs.
The maximum sequence length and embedding size have a direct impact on representation capability and inference efficiency,
so we will discuss them separately in Section~\ref{section-other-factors}.

\subsection{Overall Performance Evaluation}

In this part, 
we fulfill our primary objective by presenting the comparison 
between LLM-based recommenders and conventional recommenders.
Considering the expensive training cost of LLM-based models, 
their results reported in Table~\ref{table-overall-comparison} 
are based on a single run, 
while the results of other conventional methods are averaged over 5 independent runs.

\textbf{Conventional methods using cross-entropy outperform LLM-based recommenders.}
Let us focus on the first three blocks in Table \ref{table-overall-comparison},
where the use of cross-entropy greatly improves the conventional methods' recommendation performance.
In particular, SASRec and FMLP-Rec demonstrate superior performance compared to LLMs,
but fall significantly behind if replacing cross-entropy with BCE.
Hence, previous affirmative arguments about the LLMs' recommendation performance 
are rooted in unobjective and unfair comparisons, 
wherein BCE or BPR are commonly used for training conventional models.
Moreover, the inflated model size (from 60M of P5 to 7B of LlamaRec) only yields negligible improvements in some of the metrics.
The rich world knowledge and powerful reasoning ability seem to be of limited use here 
due to the emphasis on personalization in sequential recommendation.
In conclusion, despite after fine-tuning, 
LLM-based recommenders still fail to surpass state-of-the-art conventional models.

\textbf{Comparable effectiveness can be achieved using practical approximations.}
Conducting a full softmax over all items for cross-entropy may be infeasible in practice.
Fortunately, the last two blocks in Table \ref{table-overall-comparison}
shows the substitutability of cross-entropy by applying the modified NCE or the proposed SCE.
To showcase the effectiveness of these substitutes,
we intentionally sample a rather conservative number of negative samples,
and thus there remains a slight gap compared to cross-entropy.
Nevertheless, the superior results of NCE and SCE 
re-emphasize the clear gap that exists between LLM-based and conventional recommenders.

\begin{figure}
	\centering
    \subfloat[HR@10]{
        \includegraphics[width=0.45\textwidth]{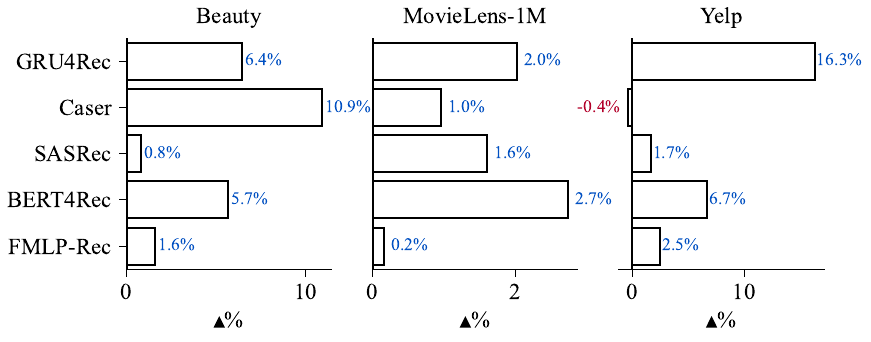}
    }
    \hfil
    \subfloat[NDCG@10]{
        \includegraphics[width=0.45\textwidth]{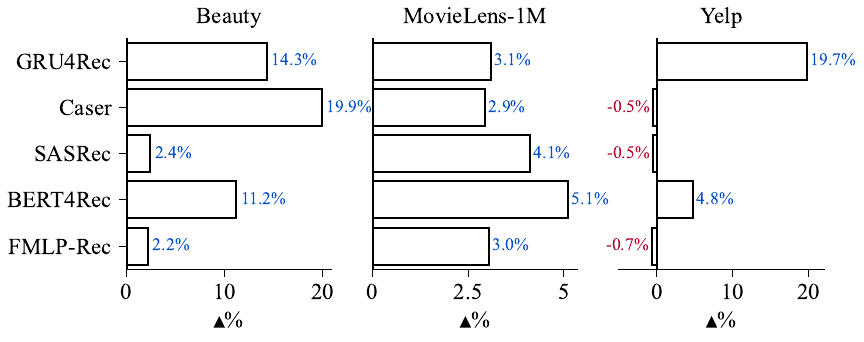}
    }
	\caption{
        Relative gaps between SCE and NCE.
	}
	\label{fig-sce-nce-gaps}
\end{figure}

\textbf{
SCE is more consistent and reliable than NCE.
}
As discussed in Section \ref{section-nce}, 
NCE is sensitive to the choice of $c$:
extremely small or large values might impede learning and degrade performance.
The inconsistent performance gains from NCE can also verify this conclusion.
Figure \ref{fig-sce-nce-gaps} clearly demonstrates that 
NCE can contribute competitive performance to SASRec and FMLP-Rec,
but underperforms SCE across a variety of models (\eg, GRU4Rec and BERT4Rec) 
and datasets (\eg, Beauty and MovieLens-1M).
For Caser on Beauty and GRU4Rec on Yelp, 
replacing SCE with NCE results in a performance degradation of even $\ge$15\%.

Additional experiments on BPR loss and MRR metrics are shown in the Appendix. 
They draw the same conclusions as above.

\subsection{Other Factors for Objective Evaluation}
\label{section-other-factors}

Due to the difference in model design, 
it is challenging to conduct evaluations on a completely identical testbed. 
To clarify the reliability of the results in Table \ref{table-overall-comparison},
we further investigate two key factors: 
maximum sequence length\footnote{
    The maximum sequence length refers to the maximum number of historical interactions 
    used for next-item prediction.
    Other tokens like prompts are not included.
} $L$
and embedding size $d$.
According to the conclusions above, SCE is employed to train SASRec in the following.
Results for CE and NCE can be found in Appendix~\ref{section-max-seq-len}.

\begin{figure}[t]
	\centering
	\includegraphics[width=0.47\textwidth]{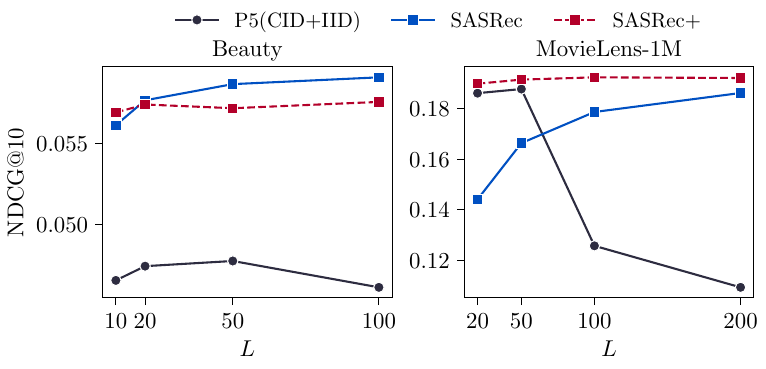}
	\caption{
        P5 (CID+IID) versus SASRec(+) across various maximum sequence length $L$.
        SASRec+ uses a sliding window similar to P5 to augment each sequence.
	}
	\label{fig-maxlen}
\end{figure}

\begin{table}[]
    \centering
    \caption{
        The impact of embedding size $d$.
    }
    \label{table-embedding-size}
    \scalebox{0.9}{
\begin{tabular}{lc|cc|cc}
    \toprule
            &                & \multicolumn{2}{c|}{Beauty}        & \multicolumn{2}{c}{MovieLens-1M}  \\
    \midrule
            & $d$ & HR@10           & NDCG@10         & HR@10           & NDCG@10         \\
    \midrule
P5(CID+IID) & 512            & 0.0791          & 0.0474          & 0.3131          & 0.1861          \\
    \midrule
SASRec+     & 64             & 0.0937          & 0.0574          & 0.3271          & 0.1915          \\
SASRec+     & 512            & \textbf{0.0963} & \textbf{0.0589} & \textbf{0.3360} & \textbf{0.1999} \\
    \bottomrule
\end{tabular}
    }
\end{table}

\textbf{Maximum sequence length.}
Following the routine of previous studies \cite{Seq:SASRec:Kang:2018,Seq:Bert4Rec:Sun:2019},
conventional models like SASRec are allowed to make predictions based on the last $L=200$ interactions on MovieLens-1M and $L=50$ on other datasets. 
In contrast, LLM-based recommenders are confined to $L=20$ on all three datasets for training efficiency.
We argue that this difference does not make the conclusion differ 
because as can be seen in Figure \ref{fig-maxlen}, P5 does not exhibit much better performance with access to more historical interactions.
Notably, SASRec's performance is stable on Beauty, but on MovieLens-1M, it deteriorates significantly when $L$ gets smaller.
This phenomenon primarily arises from the fact that 
the original implementation of SASRec only utilizes the most recent $L$ interactions for training, 
whereas for P5 each sequence is divided into multiple segments.
Consequently, P5 is able to access far more interactions than SASRec during training,
especially on the MovieLens-1M dataset known for long sequence length.
If we apply a similar strategy that augments each sequence using a sliding window,
the resulting SASRec+ then performs consistently across diverse $L$.

\textbf{Embedding size.}
Models with a larger embedding size have stronger representation capability 
and thus potentially better recommendation performance. 
According to the discussion above, 
we examine its impact under the same setting of $L=20$.
In Table~\ref{table-embedding-size}, 
when the embedding size of SASRec+ is increased from 64 to 512,
the obtained performance gains are marginal.
In view of its costly overhead, such improvement is not attractive in practice.
This also implies that existing LLM-based recommenders are over-parameterized in terms of ranking capability.

\section{Conclusion}

In this work, we bridge the theoretical and empirical performance gaps between cross-entropy and some of its approximations
through a modified noise contrastive estimation loss and an effective scaled cross-entropy loss.
Based on these practical approximations, we showcase that existing LLM-based recommenders are not as effective as claimed.
The innovative understandings and extensive experiments 
can be expected to facilitate an objective evaluation of LLM-based recommendation in the future.

\clearpage

\bibliographystyle{ACM-Reference-Format}
\bibliography{refs}

\appendix

\tableofcontents

\section{Experimental setup}

\textbf{Baselines.}
Although LLM itself has surprising zero-shot recommendation ability,
there still exist non-negligible gaps \cite{LLM4Rec:Prefrence:Kang:2023,LLM4Rec:TallRec:Bao:2023} unless domain-specific knowledge is injected.
Hence, only LLM-based recommenders enhanced by fine-tuning will be compared in this paper:
\begin{itemize}[leftmargin=*]
    \item \textbf{P5 (CID+IID)}~\cite{LLM4Rec:CIDIID:Hua2023} 
    unifies multiple tasks (\eg, sequential recommendation and rating prediction) into a sequence-to-sequence paradigm.
    The use of collaborative and independent indexing together creates LLM-compatible item IDs.
    \item \textbf{POD}~\cite{LLM4Rec:POD:Li2023} 
    bridges IDs and words by distilling long discrete prompts into a few continuous prompts.
    It also suggests a task-alternated training strategy for efficiency.
    \item \textbf{LlamaRec}~\cite{LLM4Rec:LLaMARec:Yue2023}
    aims to address the slow inference process caused by autoregressive generation.
    Given the candidates retrieved by conventional models, it subsequently reranks them based on the foundation model of Llama2.
    \item \textbf{E4SRec}~\cite{LLM4Rec:E4SRec:Li:2023}
    incorporates ID embeddings
    trained by conventional sequential models through a linear adaptor,
    and applies LORA \cite{LLM:LORA:Hu:2021} for parameter-efficient fine-tuning.
\end{itemize}
Additionally, 
five sequential models, covering MLP, CNN, RNN, and Transformer architectures,
are considered here to
uncover the true capability of conventional methods.
\begin{itemize}[leftmargin=*]
    \item \textbf{GRU4Rec} \cite{Seq:GRU4Rec:Hidasi:2015} 
    applies RNN to recommendation with specific modifications made to cope with data sparsity.
    \item \textbf{Caser} \cite{Seq:Caser:Tang:2018} treats the embedding matrix as an `image',
    and captures local patterns by utilizing conventional filters.
    \item \textbf{SASRec} \cite{Seq:SASRec:Kang:2018} and \textbf{BERT4Rec} \cite{Seq:Bert4Rec:Sun:2019} 
    are two pioneering works equipped with unidirectional and bidirectional self-attention, respectively.
    By their nature, SASRec predicts the next item based on previously interacted items, 
    while BERT4Rec is optimized through a cloze task.
    \item \textbf{FMLP-Rec} \cite{Seq:FMLPRec:Zhou:2022} denoises item sequences in the frequency domain.
    Although FMLP-Rec consists solely of MLPs, it exhibits superior performance compared to Transformer-based models.
\end{itemize}
An overview of the model statistics can be found in Table \ref{table-model-statistics}.
Note that for LlamaRec and E4SRec we employ Llama2-7B instead of Llama2-13B 
as the foundation model for training efficiency.
This results in minor performance differences.

\textbf{Implementation details.}
The implementation of LLM-based recommenders is due to their source code,
while for conventional models the code is available at \url{https://anonymous.4open.science/r/1025}.
Due to the difference in model design, 
it is challenging to conduct evaluations on a completely identical testbed.
Therefore, we follow the routine of previous studies \cite{Seq:SASRec:Kang:2018,Seq:Bert4Rec:Sun:2019},
where the maximum sequence length $L=200$ on MovieLens-1M and $L=50$ on Beauty and Yelp.
In contrast, for LLM-based models, $L=20$ on all three datasets.
Empirically, this difference does not make the conclusion differ.

\textbf{Training strategies.}
For SASRec and BERT4Rec, each item sequence is trained once per epoch.
As a result, only the most recent $L$ historical interactions are accessed during the training process, 
which will negatively impact performance when the sequence lengths are typically longer than $L$.
Some approaches such as LLM-based approaches and other conventional methods 
thus divides each sequence into multiple sub-sequences.

\section{Overview of Loss Function}
\label{section-loss-function}

\textbf{Cross-Entropy (CE)},
also known as the negative log-likelihood (NLL) loss, can be formulated as follows:
\begin{align*}
    \ell_{\text{CE}} 
    &= -\log \frac{\exp(s_{v_+})}{ \sum_{v \in \mathcal{I}} \exp(s_v)}
    = -s_{v_+} + \log \underbrace{\sum_{v \in \mathcal{I}} \exp(s_v)}_{=: Z_{\text{CE}}}.
\end{align*}
This is also the de facto objective commonly used in the pre-training (fine-tuning) of LLMs.

\textbf{Binary Cross-Entropy (BCE).} 
BCE samples one negative item $v_-$ for each target $v_+$,
which necessitates the recommender to possess excellent pointwise scoring capability:
\begin{align*}
    \ell_{\text{BCE}}
    &= -\log \sigma(s_{v_+}) - \log (1 - \sigma(s_{v_-})) \\
    &= -\log \frac{\exp(s_{v_+})}{1 + \exp(s_{v_+})} 
    -\log \frac{1}{1 + \exp(s_{v_-})} \\
    &= -s_{v_+} + \log
    \underbrace{
    \big(
        (1 + \exp(s_{v_+}))
        (1 + \exp(s_{v_-}))
    \big)
    }_{=: Z_{\text{BCE}}}.
\end{align*}
Here $\sigma: \mathbb{R} \rightarrow [0, 1]$ denotes the sigmoid function.

\textbf{Bayesian Personalized Ranking (BPR) \cite{loss:BPR:Rendle:2012}}
also samples one negative item $v_-$ for each target $v_+$,
but it intends to maximize the probability that $v_+$ will be chosen in preference to $v_-$:
\begin{align*}
    \ell_{\text{BPR}} 
    &= -\log \sigma(s_{v+} - s_{v_-}) \\
    &= -\log \frac{\exp(s_{v_+})}{\exp(s_{v_+}) + \exp(s_{v_-})} \\
    &= s_{v_+} + \log
    \underbrace{
    \bigl(
        \exp(s_{v_+}) + \exp(s_{v_-})
    \bigr)
    }_{=: Z_{\text{BPR}}}.
\end{align*}

\textbf{Importance Sampling (IS) \cite{loss:IS:Bengio:2008,loss:IS4LM:Jean:2014}}
is a widely used technique for CE approximation.
It is capable of correcting the approximation error via a proposal distribution $Q$:
\begin{align*}
    \ell_{\text{IS}}
    &= -\log \frac{
        \exp(s_{v_+} -  \log Q(v_+) )
    }{
        \sum_{i=1}^K 
        \exp( s_{v_i} -  \log Q(v_i) )
    } \\
    &= -\log \frac{
        \exp( s_{v_+})
    }{
        \sum_{i=1}^K 
        \exp \bigl( s_{v_i} -  \log Q(v_i) + \log Q(v_+) \bigr)
    } \\
    &= -s_{v_+} + \log
    \underbrace{\sum_{i=1}^K \exp \bigl( s_{v_i} -  \log Q(v_i) + \log Q(v_+) \bigr)}_{=: Z_{\text{IS}}}.
\end{align*}
In addition to uniform distribution, 
the distribution derived from popularity metrics \cite{loss:IS4Rec:Lian:2020} is also a commonly used choice.

\textbf{Noise Contrastive Estimation (NCE) \cite{loss:NCE:Gutmann:2010,loss:NCE4LM:Mnih:2012}}
requires the model to discriminate the target $v_+$ from an easy-to-sample noise distribution:
\begin{align*}
    \ell_{\text{NCE}}
    &= -\log \sigma(s_{v_+}') - \sum_{i=1}^K \log (1 - \sigma(s_{v_i}')) \\
    &= -\log \frac{\exp(s_{v_+}')}{1 + \exp(s_{v_+}')} - \sum_{i=1}^K \log \frac{1}{1 + \exp(s_{v_i}')} \\
    &= -s_{v_+}' + \log 
    \underbrace{
    \biggl((1 + \exp(s_{v_+}')) \prod_{i=1}^K (1 + \exp(s_{v_i}') \biggr)
    }_{=: Z_{\text{NCE}}}.
\end{align*}
In the case of uniform sampling, $s_{v}' = s_v - c - \log \frac{K}{|\mathcal{I}|}$,
where $c$ is a trainable parameter as an estimate of $\log Z_{\text{CE}}$.

\textbf{NEGative sampling (NEG) \cite{loss:Word2Vec:Mikolov:2013}} is a special case of NCE by fixing $c = \log \frac{|\mathcal{I}|}{K}$:
\begin{align*}
    \ell_{\text{NEG}}
    &= -\log \sigma(s_{v_+}) - \sum_{i=1}^K \log (1 - \sigma(s_{v_i})) \\
    &= -\log \frac{\exp(s_{v_+})}{1 + \exp(s_{v_+})} - \sum_{i=1}^K \log \frac{1}{1 + \exp(s_{v_i})} \\
    &= -s_{v_+} + \log 
    \underbrace{
    \biggl((1 + \exp(s_{v_+})) \prod_{i=1}^K (1 + \exp(s_{v_i}) \biggr)
    }_{=: Z_{\text{NEG}}}.
\end{align*}

\section{Proofs}
\label{section-proofs}

This section completes the proofs regarding the connection between the aforementioned loss functions and ranking metrics.
Before delving into the proofs, 
it is important to note that in this paper
we focus on the metrics specifically for next-item recommendation 
as it is most consistent with LLM's next-token generation feature.

\subsection{Proof of Proposition \ref{proposition-bounding}}

\begin{proposition}
    For a target item $v_+$ which is ranked as $r_+$, the following inequality holds true for any $n \ge r_+$
    \begin{equation}
        -\log \mathrm{NDCG}(r_+) \le \ell_{\mathrm{CE}\text{-}n},
    \end{equation}
    where
    \begin{equation}
        \ell_{\mathrm{CE}\text{-}n} := s_{v_+} + \log \sum_{r(v) \le n} \exp(s_{v}).
    \end{equation}
\end{proposition}

\begin{proof}

Notice that $\log_2 (1 + x) \le x$ holds true for any $x \ge 1$. 
Hence, we have
\begin{align*}
    &\text{NDCG}(r_+) = \frac{1}{\log_2 (1 + r_+)} \\
    \ge & \frac{1}{r_+} = \frac{1}{1 + \sum_{v \not = v_+} \delta(s_{v} > s_{v_+})}
    = \frac{1}{1 + \sum_{r(v) < r_+} \delta(s_{v} > s_{v_+})} \\
    \ge & \frac{1}{1 + \sum_{r_v < r_+} \exp(s_v - s_{v_+})} 
    = \frac{\exp(s_{v_+})}{\exp(s_{v_+}) + \sum_{r(v) < r_+} \exp(s_v)} \\
    \ge & \frac{\exp(s_{v_+})}{\sum_{r(v) \le n} \exp(s_v)},
\end{align*}
where $\delta(\text{condition}) = 1$ if the given condition is true else 0,
and the second-to-last inequality holds because $\exp(s_v - s_{v_+}) \ge 1$ when $s_v > s_{v_+}$.

\end{proof}

\subsection{Proof of Theorem \ref{theorem-nce-neg-bounds} and \ref{theorem-sce-bounds}}

First, let us introduce some lemmas that give lower bounds of the loss functions.

\begin{lemma}
    \label{lemma-bound-nce}
    Let $\xi_1$ be the number of sampled items with non-negative corrected scores; that is,
    \begin{equation}
        \xi_1 = 
        \bigl|
            \{
                v_i: s_{v_i}' \ge 0, \: i=1,2,\ldots, K
            \}
        \bigr|.
    \end{equation}
    Then, we have
    \begin{equation}
        \ell_{\mathrm{NCE}} \ge \xi_1 \log 2.
    \end{equation}
\end{lemma}

\begin{proof}
    According to the definition of NCE, it follows that
    \begin{align*}
        \ell_{\text{NCE}}
        &= -s_{v_+}' + \log \biggl((1 + \exp(s_{v_+}')) \prod_{i=1}^K (1 + \exp(s_{v_i}') \biggr) \\
        &= \log \biggl((1 + \exp(-s_{v_+}')) \prod_{i=1}^K (1 + \exp(s_{v_i}') \biggr) \\
        &\ge \log \biggl(\prod_{i=1}^K (1 + \exp(s_{v_i}') \biggr) = \sum_{i=1}^K \log (1 + \exp(s_{v_i}')) \\
        &\ge \sum_{i=1}^K \delta(s_{v_i}' \ge 0) \log 2 = \xi_1 \log 2.
    \end{align*}
\end{proof}

\begin{lemma}
    \label{lemma-bound-neg}
    Let $\xi_2$ be the number of sampled items with non-negative scores; that is,
    \begin{equation}
        \xi_2 = 
        \bigl|
            \{
                v_i: s_{v_i} \ge 0, \: i=1,2,\ldots, K
            \}
        \bigr|.
    \end{equation}
    Then, we have
    \begin{equation}
        \ell_{\mathrm{NEG}} \ge \xi_2 \log 2.
    \end{equation}
\end{lemma}

\begin{proof}
    The proof is completely the same as that of Lemma \ref{lemma-bound-nce}.
\end{proof}

\begin{lemma}
    \label{lemma-bound-scaled}
    Let $\xi_3$ be the number of sampled items with scores not lower than that of the target; that is
    \begin{equation}
        \xi_3 = 
        \bigl|
            \{
                v_i: s_{v_i} \ge s_{v_+}, \: i=1,2,\ldots, K
            \}
        \bigr|.
    \end{equation}
    Then, we have
    \begin{align}
        \ell_{\mathrm{SCE}} \ge \log (1 + \alpha \xi_3).
    \end{align}
\end{lemma}

\begin{proof}
    According to the definition of SCE, it follows that
    \begin{align*}
        \ell_{\text{SCE}}
        &= -s_{v_+} + \log 
        \bigl(
            \exp(s_{v_+}) + \alpha \sum_{i=1}^K \exp(s_{v_i})
        \bigr) \\
        &= \log 
        \bigl(
            1 + \alpha \sum_{i=1}^K \exp(s_{v_i} - s_{v_+})
        \bigr) \\
        &\ge \log 
        \bigl(
            1 + \alpha \sum_{i=1}^K \delta(s_{v_i} \ge s_{v_+}).
        \bigr) \\
        &= \log (1 + \alpha \xi_3).
    \end{align*}
\end{proof}

\begin{lemma}
    \label{lemma-bound-IS}
    Let $\xi_4$ be the number of sampled items with scores higher than that of the target; that is
    \begin{equation}
        \xi_4 = 
        \bigl|
            \{
                v_i: s_{v_i} \ge s_{v_+}, \: i=1,2,\ldots, K
            \}
        \bigr|.
    \end{equation}
    Then, we have\footnote{
        Assume that
        \begin{equation*}
            \log 0 := \lim_{\epsilon \rightarrow 0} \log \epsilon = -\infty.
        \end{equation*}
    }
    \begin{align}
        \ell_{\mathrm{IS}} \ge \log (\xi_4)
    \end{align}
    if the proposal distribution $Q(v) \equiv 1 / |\mathcal{I}|$.
\end{lemma}

\begin{proof}
    According to the definition of importance sampling, it follows that
    \begin{align*}
        \ell_{\text{IS}}
        &= -s_{v_+} + \log 
        \bigl(
            \sum_{i=1}^K \exp(s_{v_i})
        \bigr) \\
        &= \log 
        \bigl(
            \sum_{i=1}^K \exp(s_{v_i} - s_{v_+})
        \bigr) \\
        &\ge \log 
        \bigl(
            \sum_{i=1}^K \delta(s_{v_i} \ge s_{v_+}).
        \bigr) \\
        &= \log (\xi_4).
    \end{align*}
\end{proof}

\begin{lemma}
    \label{lemma-binomial-bound}
    Let $\xi \sim \mathcal{B}(K, p)$ denote a random variable representing the number of successes over $K$ binomial trials with a probability of $p$.
    Then, we have
    \begin{align}
        \mathbb{P}(\xi \ge m) \ge 1 - m (1 - p)^{\lfloor K / m \rfloor}, \quad \forall m=0, 1, \ldots, K.
    \end{align}
\end{lemma}

\begin{proof}\footnote{The proof follows from the response posted at the URL: \url{https://math.stackexchange.com/questions/3626472/upper-bound-on-binomial-distribution}.}
    Divide the $K$ independent binomial trials into $m$ disjoint groups, 
    each containing at least $\lfloor K/m \rfloor$ trials.
    If $\xi < m$, then one of the groups must have no successes observed; formally, we have
    \begin{align}
        \mathbb{P}(\xi < m)
        &\le \mathbb{P}(
            \bigcup_{i=1}^m \{ \text{no successes observed in group } i \} 
        ) \\
        &\le \sum_{i=1}^m \mathbb{P}(
            \{ \text{no successes observed in group } i \} 
        ) \\ 
        &\le m (1 - p)^{\lfloor K/m \rfloor}.
    \end{align}
    Hence, the proof is completed by noting the fact that
    \begin{align}
        \mathbb{P}(\xi \ge m) = 1 - \mathbb{P}(\xi < m).
    \end{align}
\end{proof}

\begin{theorem}
    \label{theorem-bounds-ndcg}
    Let $v_+$ be a target item which is ranked as $r_+ \le 2^{2^m} - 1$ for some $m \in \mathbb{N}$,
    and
    \begin{align*}
        \mathcal{S}_+ := \{v \in \mathcal{I}: s_{v} \ge 0\},
        \quad 
        \mathcal{S}_+' := \{v \in \mathcal{I}: s_{v}' \ge 0\}.
    \end{align*}
    If we uniformly sample $K$ items for training, then with a probability of at least
    \begin{equation}
        \left \{
        \begin{array}{ll}
            1 - m\big(1 - |\mathcal{S}_+'| / |\mathcal{I}| \big)^{\lfloor K/m \rfloor}, & \text{ if } \ell_* = \ell_{\mathrm{NCE}} \\
            1 - m\big(1 - |\mathcal{S}_+| / |\mathcal{I}| \big)^{\lfloor K/m \rfloor}, & \text{ if } \ell_* = \ell_{\mathrm{NEG}} \\
            1 -\frac{1}{\alpha}2^m\big(1 - r_+ / |\mathcal{I}| \big)^{\lfloor \alpha K / 2^m \rfloor}, & \text{ if } \ell_* = \ell_{\mathrm{SCE}} \\
            1 -2^m\big(1 - r_+ / |\mathcal{I}| \big)^{\lfloor K / 2^m \rfloor}, & \text{ if } \ell_* = \ell_{\mathrm{IS}} \\
        \end{array}
        \right .,
    \end{equation}
    we have
    \begin{equation}
        \label{eq-bound-required}
        -\log \mathrm{NDCG}(r_+) \le \ell_*.
    \end{equation}
\end{theorem}

\begin{proof}

As $r_+ \le 2^{2^m} - 1$ for some $m \in \mathbb{N}$, we immediately have
\begin{equation*}
    -\log \text{NDCG}(r_+) \le  m \log 2.
\end{equation*}
Now, let us prove the conclusions one by one.

\begin{itemize}
    \item[Case 1.] According to Lemma \ref{lemma-bound-nce} we know that
    \begin{equation*}
        \ell_{\text{NCE}} \ge \xi_1 \log 2.
    \end{equation*}
    Therefore, Eq. \eqref{eq-bound-required} holds true for NCE as long as $\xi_1 \ge m$.
    Formally, we have
    \begin{equation}
        \mathbb{P}
        \bigl(
            -\log \mathrm{NDCG}(r_+) \le \ell_{\text{NCE}}
        \bigr)
        \ge
        \mathbb{P}
        \bigl(
            \xi_1 \ge m
        \bigr).
    \end{equation}
    Also notice that uniformly sampling from $\mathcal{I}$ yields a probability of $p=|\mathcal{S}_+'| / |\mathcal{I}|$ such that
    the corrected score of the sampled item is non-negative.
    Therefore, based on Lemma~\ref{lemma-binomial-bound}, we have
    \begin{equation}
        \mathbb{P}
        \bigl(
            \xi_1 \ge m
        \bigr) \ge 1 - m (1 - |\mathcal{S}_+'| / |\mathcal{I}|)^{\lfloor K / m \rfloor}.
    \end{equation}

    \item[Case 2.] The proof of NEG is completely the same as that of NCE.

    \item[Case 3.] Analogously, Lemma \ref{lemma-bound-scaled} implies that
    \begin{align}
        & \mathbb{P}
        \bigl(
            -\log \mathrm{NDCG} \le \ell_{\text{SCE}}
        \bigr) \\
        \ge &
        \mathbb{P}
        \bigl(
            \xi_3 \ge \frac{2^m - 1}{\alpha}
        \bigr) \\
        \ge &
        \mathbb{P}
        \bigl(
            \xi_3 \ge \frac{2^m}{\alpha}
        \bigr)
        =
        \mathbb{P}
        \bigl(
            \xi_3 \ge \bigl\lfloor \frac{2^m}{\alpha} \bigr\rfloor
        \bigr).
    \end{align}
    Also notice that uniformly sampling from $\mathcal{I}$ yields a probability of $p=r_+ / |\mathcal{I}|$ such that
    the score of the sampled item is not lower than that of the target (\ie, the top-$r_+$ ranked items).
    Therefore, based on Lemma \ref{lemma-binomial-bound}, we have
    \begin{align}
        \mathbb{P}
        \bigl(
            \xi_3 < \bigl\lfloor \frac{2^m}{\alpha} \bigr\rfloor
        \bigr)
        &\le \bigl\lfloor \frac{2^m}{\alpha} \bigr\rfloor (1 - r_+ / |\mathcal{I}_+|)^{\lfloor K/ \lfloor \frac{2^m}{\alpha} \rfloor \rfloor } \\
        &\le \frac{2^m}{\alpha} (1 - r_+ / |\mathcal{I}_+|)^{\lfloor K/ \frac{2^m}{\alpha} \rfloor }.
    \end{align}
    \item[Case 3.] The proof of importance sampling is similar to that of SCE.

\end{itemize}
    The proof has been completed now.

\end{proof}

\subsection{Proof of Eq. \eqref{eq-scaled-IS}}

Here, we show that $\ell_{\text{SCE}}$ is morphologically equivalent to $\ell_{\text{IS}}$ if
\begin{equation*}
    \label{eq-scaled-IS-2}
    Q(v) = 
    \left \{
    \begin{array}{ll}
        \frac{
            \alpha
        }{
            |\mathcal{I}| - 1 + \alpha
        } & \text{ if } v = v_+ \\
        \frac{
            1
        }{
            |\mathcal{I}| - 1 + \alpha
        } & \text{ if } v \not= v_+
    \end{array}
    \right..
\end{equation*}

\begin{proof}

Under the condition of Eq. \eqref{eq-scaled-IS-2}, we have
\begin{align*}
    & \ell_{\text{SCE}} \\
    =& -\log \frac{
        \exp(s_{v_+})
    }{
        \exp(s_{v_+}) + \alpha \sum_{i=1}^K \exp( s_{v_i})
    } \\
    =&-\log \frac{
        \frac{|\mathcal{I}|-1 + \alpha}{\alpha} \exp(s_{v_+})
    }{
        \frac{|\mathcal{I}|-1 + \alpha}{\alpha} \exp(s_{v_+}) + \frac{|\mathcal{I}|-1 + \alpha}{1}  \sum_{i=1}^K \exp( s_{v_i})
    } \\
    =&-\log \frac{
        \exp \biggl(s_{v_+} - \log \frac{\alpha}{|\mathcal{I}|-1 + \alpha} \biggr)
    }{
        \exp \biggl(s_{v_+} - \log \frac{\alpha}{|\mathcal{I}|-1 + \alpha} \biggr)
         + \sum_{i=1}^K \exp \biggl( s_{v_i} - \log \frac{1}{|\mathcal{I}|-1 + \alpha} \biggr)
    },
\end{align*}
    which is morphologically equivalent to $\ell_{\text{IS}}$ with $K+1$ items in the normalizing term.

\end{proof}

\subsection{Proofs Regarding Reciprocal Rank (RR)}

Reciprocal Rank (RR) is another popular metric used to measure the ranking capability,
which is defined as follows
\begin{equation}
    \text{RR}(r_+) = \frac{1}{r_+}.
\end{equation}
We provide some theoretical conclusions here and leave the empirical results in next section.
Firstly, we connect RR to the cross-entropy as Proposition \ref{proposition-bounding} does for NDCG.

\begin{proposition}
    \label{proposition-tighter-bounds-mrr}
    For a target item $v_+$ which is ranked as $r_+$, the following inequality holds true for any $n \ge r_+$
    \begin{equation}
        -\log \mathrm{MRR}(r_+) \le \ell_{\mathrm{CE}\text{-}n},
    \end{equation}
    where
    \begin{equation*}
        \ell_{\mathrm{CE}\text{-}n} := s_{v_+} + \log \sum_{r(v) \le n} \exp(s_{v}).
    \end{equation*}
\end{proposition}

\begin{proof}
    \begin{align*}
        &\text{RR}(r_+) = \frac{1}{r_+} \\
        =& \frac{1}{1 + \sum_{v \not = v_+} \delta(s_{v} > s_{v_+})}
        = \frac{1}{1 + \sum_{r(v) < r_+} \delta(s_{v} > s_{v_+})} \\
        \ge & \frac{1}{1 + \sum_{r_v < r_+} \exp(s_v - s_{v_+})} 
        = \frac{\exp(s_{v_+})}{\exp(s_{v_+}) + \sum_{r(v) < r_+} \exp(s_v)} \\
        \ge & \frac{\exp(s_{v_+})}{\sum_{r(v) \le n} \exp(s_v)}.
    \end{align*}
\end{proof}

Next, we establish a connection between RR and NCE, NEG, SCE, and importance sampling, 
similar to what Theorem \ref{theorem-bounds-ndcg} does for NDCG.
\begin{theorem}
    \label{theorem-bounds-mrr}
    Let $v_+$ be a target item which is ranked as $r_+ \le 2^m$ for some $m \in \mathbb{N}$,
    and
    \begin{align*}
        \mathcal{S}_+ := \{v \in \mathcal{I}: s_{v} \ge 0\},
        \quad 
        \mathcal{S}_+' := \{v \in \mathcal{I}: s_{v}' \ge 0\}.
    \end{align*}
    If we uniformly sample $K$ items for training, then with a probability of at least
    \begin{equation}
        \left \{
        \begin{array}{ll}
            1 - m\big(1 - |\mathcal{S}_+'| / |\mathcal{I}| \big)^{\lfloor K/m \rfloor}, & \text{ if } \ell_* = \ell_{\mathrm{NCE}} \\
            1 - m\big(1 - |\mathcal{S}_+| / |\mathcal{I}| \big)^{\lfloor K/m \rfloor}, & \text{ if } \ell_* = \ell_{\mathrm{NEG}} \\
            1 -\frac{1}{\alpha}2^m\big(1 - r_+ / |\mathcal{I}| \big)^{\lfloor \alpha K / 2^m \rfloor}, & \text{ if } \ell_* = \ell_{\mathrm{SCE}} \\
            1 -2^m\big(1 - r_+ / |\mathcal{I}| \big)^{\lfloor K / 2^m \rfloor}, & \text{ if } \ell_* = \ell_{\mathrm{IS}} \\
        \end{array}
        \right .,
    \end{equation}
    we have
    \begin{equation}
        -\log \mathrm{RR}(r_+) \le \ell_*.
    \end{equation}
\end{theorem}

\begin{proof}

As $r_+ \le 2^m$ for some $m \in \mathbb{N}$, we immediately have
\begin{equation*}
    -\log \text{RR}(r_+) \le  m \log 2.
\end{equation*}
The conclusions then can be proved in the same way as Theorem~\ref{theorem-bounds-ndcg}.

\end{proof}

\begin{remark}
    It is worth noting that the inequality for RR is achieved at a stricter condition compared to NDCG.
    For the same rank $r_+$, NDCG allows for a smaller value of $m$, 
    thereby yielding slightly higher bounding probabilities.
    However, this nuance does not undermine the fact that the same conclusions can be drawn from the two metrics.
    The empirical observations in the next section can be used to verify this.
\end{remark}

\section{Empirical results on Reciprocal Rank (RR)}
\label{section-mrr}

\begin{table}
    \caption{
        MRR@5 and MRR@10 comparison on the Beauty, MovieLens-1M, and Yelp datasets.
        The best results of each block are marked in bold.
        `\improv over CE/LLM' represents the relative gap between the respective best results.
    }
    \label{table-overall-mrr}
    \centering
    \scalebox{0.7}{
\begin{tabular}{c|l|cc|cc|cc}
    \toprule
\multicolumn{1}{l}{}  &              & \multicolumn{2}{c}{Beauty}        & \multicolumn{2}{c}{MovieLens-1M}  & \multicolumn{2}{c}{Yelp}          \\
    \midrule
\multicolumn{1}{l}{}  &              & MRR@5           & MRR@10          & MRR@5           & MRR@10          & MRR@5           & MRR@10          \\
    \midrule
    \midrule
\multirow{4}{*}{LLM}  & POD          & 0.0107          & 0.0116          & 0.0249          & 0.0264          & 0.0267          & 0.0292          \\
                      & P5(CID+IID)  & \textbf{0.0345} & 0.0376          & \textbf{0.1312} & \textbf{0.1429} & 0.0164          & 0.0186          \\
                      & LlamaRec     & 0.0344          & \textbf{0.0380} & 0.0903          & 0.1046          & \textbf{0.0270} & \textbf{0.0295} \\
                      & E4SRec       & 0.0326          & 0.0356          & 0.1025          & 0.1144          & 0.0174          & 0.0196          \\
    \midrule
\multirow{5}{*}{CE}   & GRU4Rec      & 0.0281          & 0.0310          & 0.1311          & 0.1438          & 0.0137          & 0.0161          \\
                      & Caser        & 0.0260          & 0.0284          & 0.1302          & 0.1417          & 0.0187          & 0.0200          \\
                      & SASRec       & \textbf{0.0443} & \textbf{0.0479} & 0.1287          & 0.1408          & 0.0302          & 0.0331          \\
                      & BERT4Rec     & 0.0294          & 0.0325          & 0.1118          & 0.1242          & 0.0207          & 0.0231          \\
                      & FMLP-Rec     & 0.0438          & 0.0473          & \textbf{0.1354} & \textbf{0.1481} & \textbf{0.0316} & \textbf{0.0348} \\
\multicolumn{2}{c}{\improv over LLM} & 28.1\%          & 26.0\%          & 3.2\%           & 3.6\%           & 16.9\%          & 18.1\%          \\
    \midrule
    \midrule
\multirow{5}{*}{BPR}  & GRU4Rec      & 0.0125          & 0.0147          & 0.0862          & 0.0977          & 0.0086          & 0.0102          \\
                      & Caser        & 0.0143          & 0.0163          & 0.1016          & 0.1132          & 0.0185          & 0.0201          \\
                      & SASRec       & 0.0248          & 0.0279          & 0.0934          & 0.1045          & 0.0137          & 0.0158          \\
                      & BERT4Rec     & 0.0121          & 0.0141          & 0.0690          & 0.0795          & 0.0119          & 0.0138          \\
                      & FMLP-Rec     & \textbf{0.0269}          & \textbf{0.0298} & \textbf{0.1034} & \textbf{0.1153} & \textbf{0.0286} & \textbf{0.0314} \\
\multicolumn{2}{c}{\improv over CE}  & -39.2\%         & -37.7\%         & -23.6\%         & -22.1\%         & -9.6\%          & -9.8\%          \\
\multicolumn{2}{c}{\improv over LLM} & -22.1\%         & -21.5\%         & -21.2\%         & -19.3\%         & 5.7\%           & 6.5\%           \\
    \midrule
\multirow{5}{*}{BCE}  & GRU4Rec      & 0.0108          & 0.0129          & 0.0836          & 0.0953          & 0.0079          & 0.0094          \\
                      & Caser        & 0.0154          & 0.0174          & 0.0894          & 0.1005          & 0.0198          & 0.0214          \\
                      & SASRec       & 0.0225          & 0.0257          & 0.0859          & 0.0977          & 0.0192          & 0.0215          \\
                      & BERT4Rec     & 0.0122          & 0.0144          & 0.0641          & 0.0743          & 0.0111          & 0.0131          \\
                      & FMLP-Rec     & \textbf{0.0248} & \textbf{0.0281} & \textbf{0.0967} & \textbf{0.1088} & \textbf{0.0287} & \textbf{0.0312} \\
\multicolumn{2}{c}{\improv over CE}  & -43.9\%         & -41.2\%         & -28.6\%         & -26.5\%         & -9.2\%          & -10.5\%         \\
\multicolumn{2}{c}{\improv over LLM} & -28.1\%         & -26.0\%         & -26.3\%         & -23.8\%         & 6.1\%           & 5.7\%           \\
    \midrule
\multirow{5}{*}{NCE}  & GRU4Rec      & 0.0241          & 0.0270          & 0.1300          & 0.1420          & 0.0118          & 0.0141          \\
                      & Caser        & 0.0213          & 0.0238          & 0.1296          & 0.1415          & 0.0194          & 0.0209          \\
                      & SASRec       & 0.0419          & 0.0455          & 0.1249          & 0.1377          & 0.0303          & 0.0330          \\
                      & BERT4Rec     & 0.0271          & 0.0303          & 0.1098          & 0.1226          & 0.0232          & 0.0257          \\
                      & FMLP-Rec     & \textbf{0.0425} & \textbf{0.0460} & \textbf{0.1329} & \textbf{0.1460} & \textbf{0.0316} & \textbf{0.0348} \\
\multicolumn{2}{c}{\improv over CE}  & -4.0\%          & -3.8\%          & -1.8\%          & -1.4\%          & -0.1\%          & 0.0\%           \\
\multicolumn{2}{c}{\improv over LLM} & 23.0\%          & 21.2\%          & 1.3\%           & 2.2\%           & 16.8\%          & 18.1\%          \\
    \midrule
\multirow{5}{*}{SCE}    & GRU4Rec    & 0.0296          & 0.0323          & 0.1349          & 0.1475          & 0.0148          & 0.0173          \\
                        & Caser      & 0.0277          & 0.0300          & 0.1360          & 0.1474          & 0.0193          & 0.0208          \\
                        & SASRec     & 0.0435          & 0.0470          & 0.1335          & 0.1457          & 0.0295          & 0.0324          \\
                        & BERT4Rec   & 0.0318          & 0.0349          & 0.1184          & 0.1308          & 0.0239          & 0.0265          \\
                        & FMLP-Rec   & \textbf{0.0436} & \textbf{0.0472} & \textbf{0.1410} & \textbf{0.1531} & \textbf{0.0304} & \textbf{0.0339} \\
\multicolumn{2}{c|}{\improv over CE}  & -1.5\%          & -1.3\%          & 4.2\%           & 3.4\%           & -3.7\%          & -2.8\%          \\
\multicolumn{2}{c|}{\improv over LLM} & 26.1\%          & 24.3\%          & 7.5\%           & 7.2\%           & 12.6\%          & 14.9\%          \\
    \bottomrule
\end{tabular}
    }
\end{table}

\begin{figure}[t]
	\centering
	\includegraphics[width=0.47\textwidth]{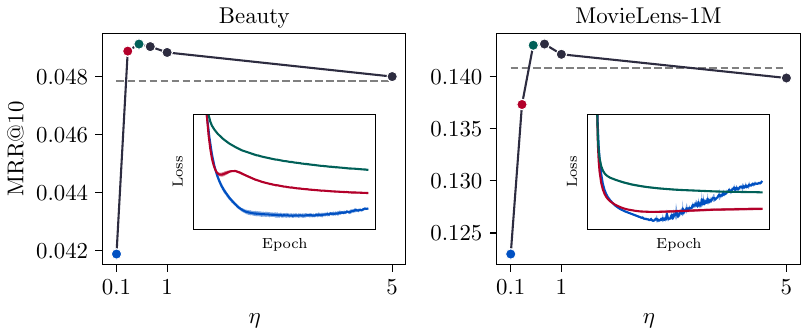}
	\caption{
        Performance comparison based on tighter bounds for MRR.
        The dashed line represents the results trained by CE (namely the case of $\eta \rightarrow +\infty$).
	}
	\label{fig-tighter-bounds-mrr}
\end{figure}

\begin{figure}[t]
	\centering
	\includegraphics[width=0.47\textwidth]{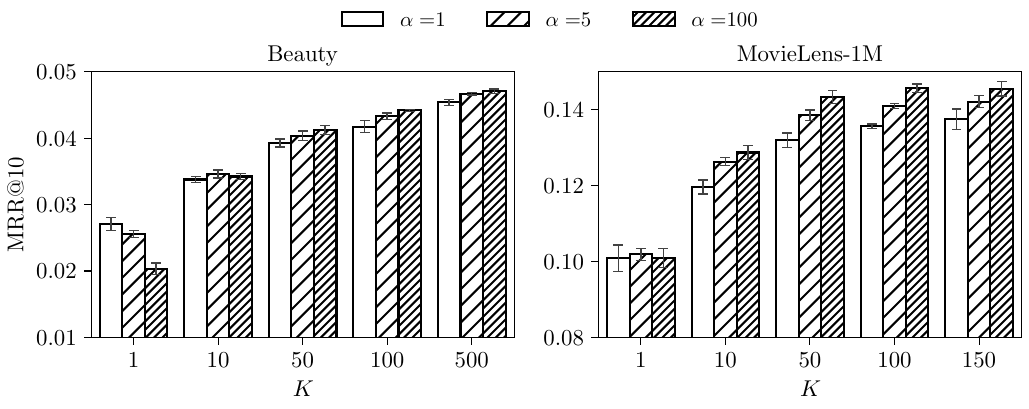}
	\caption{
        MRR@10 performance under various weight $\alpha$.
	}
	\label{fig-sce-alpha-negs-mrr}
\end{figure}

\begin{figure*}
	\centering
    \subfloat[Beauty]{
        \includegraphics[width=0.46\textwidth]{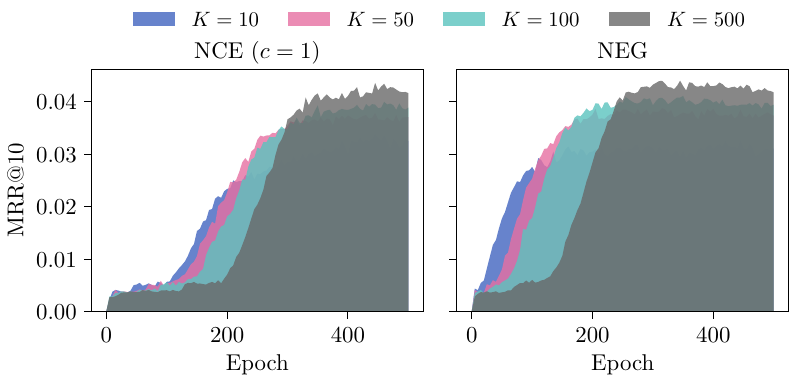}
    }
    \subfloat[MovieLens-1M]{
        \includegraphics[width=0.46\textwidth]{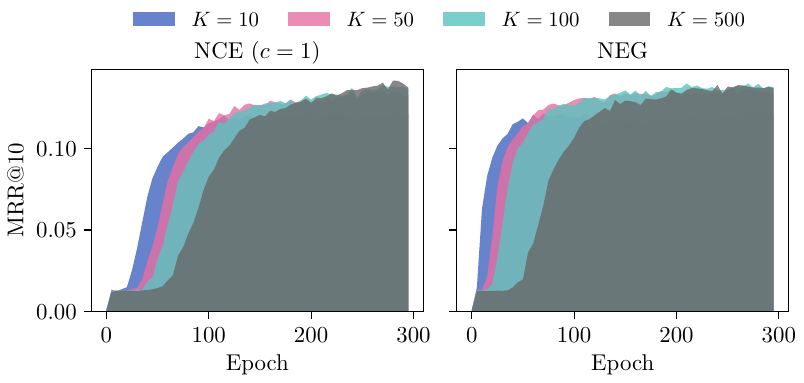}
    }
	\caption{
        MRR@10 performance of NCE and NEG across different number of negative samples.
	}
	\label{fig-nce-neg-negs-mrr}
\end{figure*}

In this part, we provide the empirical results on Reciprocal Rank (RR),
which are very similar to those of NDCG.
Note that Mean Reciprocal Rank (MRR) reported below represents the average performance over all users.

According to Proposition \ref{proposition-tighter-bounds-mrr},
$-\log \text{RR}(r_+)$ would be strictly bounded by CE-like losses,
as long as all items ranked before $v_+$ are retained in the normalizing term.
Following the strategy of NDCG, we perform the adaptive truncation Eq. \eqref{eq-truncation} to investigate the effect of a tighter bound on RR.
Figure \ref{fig-tighter-bounds-mrr} depicts a similar curve to NDCG:
It increases as the training instability is gradually overcome, 
and decreases as the objective approaches cross-entropy (\ie, $\eta \rightarrow +\infty$).

In addition, 
Figure \ref{fig-sce-alpha-negs-mrr} and Figure \ref{fig-nce-neg-negs-mrr}
respectively demonstrate the effectiveness of SCE and the limitations of NCE and NEG.
By sampling more negative samples per iteration,
SCE enjoys superior performance, 
whereas both NCE and NEG suffer from the training difficulties caused by the weak bounds in the early training stages.
Thus, SCE is arguably a simpler and preferable approximation to cross-entropy.

Finally, Table~\ref{table-overall-mrr} reports the overall comparisons w.r.t. MRR,
the corresponding conclusions drawn from Table~\ref{table-overall-comparison} can also be observed here.

\section{Maximum Sequence Length}
\label{section-max-seq-len}

\begin{table*}
    \centering
    \caption{
        Recommendation performance under various maximum sequence length $L$.
        The best results of each block are marked in bold.
        In contrast to SASRec,
        SASRec+ augments each sequence using a sliding window,
        so all interactions will be learned during training.
    }
    \label{table-seq-length}
\begin{tabular}{c|c|ccccc|ccccc}
    \toprule
\multicolumn{1}{l}{} & \multicolumn{1}{l}{}         & \multicolumn{1}{|l}{} & \multicolumn{4}{c|}{Beauty}                                            &     & \multicolumn{4}{c}{MovieLens-1M}                                      \\
\multicolumn{1}{l}{} & Method                       & $L$                    & HR@5            & HR@10           & NDCG@5          & NDCG@10         & $L$   & HR@5            & HR@10           & NDCG@5          & NDCG@10         \\
    \midrule
\multirow{4}{*}{LLM} & \multirow{4}{*}{P5(CID+IID)} & 10                   & 0.0554          & 0.0785          & 0.0391          & 0.0466          & 20  & 0.2225          & 0.3131          & 0.1570          & 0.1861          \\
                     &                              & 20                   & \textbf{0.0569} & \textbf{0.0791} & 0.0403          & 0.0474          & 50  & \textbf{0.2300} & \textbf{0.3204} & \textbf{0.1585} & \textbf{0.1877} \\
                     &                              & 50                   & 0.0559          & 0.0784          & \textbf{0.0405} & \textbf{0.0478} & 100 & 0.1560          & 0.2252          & 0.1035          & 0.1258          \\
                     &                              & 100                  & 0.0557          & 0.0776          & 0.0391          & 0.0461          & 200 & 0.1382          & 0.1949          & 0.0919          & 0.1095          \\
    \midrule
\multirow{8}{*}{CE}  & \multirow{4}{*}{SASRec}      & 10                   & 0.0673          & 0.0938          & 0.0479          & 0.0565          & 20  & 0.1715          & 0.2481          & 0.1161          & 0.1407          \\
                     &                              & 20                   & 0.0691          & 0.0973          & 0.0492          & 0.0583          & 50  & 0.1975          & 0.2827          & 0.1341          & 0.1616          \\
                     &                              & 50                   & 0.0713          & 0.0986          & 0.0510          & 0.0597          & 100 & 0.2110          & 0.3017          & 0.1446          & 0.1738          \\
                     &                              & 100                  & \textbf{0.0715} & \textbf{0.0995} & \textbf{0.0511} & \textbf{0.0601} & 200 & \textbf{0.2221} & \textbf{0.3131} & \textbf{0.1518} & \textbf{0.1812} \\
            \cmidrule{2-12}
                     & \multirow{4}{*}{SASRec+}     & 10                   & 0.0670          & 0.0933          & 0.0477          & 0.0562          & 20  & 0.2309          & 0.3233          & 0.1590          & 0.1887          \\
                     &                              & 20                   & \textbf{0.0691} & 0.0952          & 0.0490          & 0.0574          & 50  & \textbf{0.2336} & 0.3243          & \textbf{0.1614} & \textbf{0.1906} \\
                     &                              & 50                   & 0.0685          & 0.0952          & 0.0485          & 0.0571          & 100 & 0.2314          & \textbf{0.3246} & 0.1589          & 0.1890          \\
                     &                              & 100                  & 0.0688          & \textbf{0.0955} & \textbf{0.0492} & \textbf{0.0578} & 200 & 0.2295          & 0.3221          & 0.1589          & 0.1889          \\
    \midrule
\multirow{8}{*}{NCE} & \multirow{4}{*}{SASRec}      & 10                   & 0.0654          & 0.0923          & 0.0460          & 0.0547          & 20  & 0.1703          & 0.2448          & 0.1132          & 0.1372          \\
                     &                              & 20                   & 0.0676          & 0.0949          & 0.0476          & 0.0564          & 50  & 0.1905          & 0.2775          & 0.1281          & 0.1562          \\
                     &                              & 50                   & \textbf{0.0686} & \textbf{0.0961} & \textbf{0.0485} & \textbf{0.0573} & 100 & 0.2079          & 0.3001          & 0.1406          & 0.1703          \\
                     &                              & 100                  & 0.0678          & 0.0955          & 0.0478          & 0.0567          & 200 & \textbf{0.2177} & \textbf{0.3135} & \textbf{0.1479} & \textbf{0.1788} \\
            \cmidrule{2-12}
                     & \multirow{4}{*}{SASRec+}     & 10                   & 0.0649          & 0.0912          & 0.0458          & 0.0542          & 20  & 0.2267          & 0.3236          & 0.1560          & 0.1873          \\
                     &                              & 20                   & 0.0652          & \textbf{0.0915} & \textbf{0.0467} & \textbf{0.0552} & 50  & 0.2291          & 0.3217          & 0.1572          & 0.1871          \\
                     &                              & 50                   & 0.0653          & 0.0911          & 0.0462          & 0.0545          & 100 & \textbf{0.2297} & \textbf{0.3262} & \textbf{0.1574} & \textbf{0.1886} \\
                     &                              & 100                  & \textbf{0.0659} & 0.0913          & 0.0466          & 0.0548          & 200 & 0.2293          & 0.3246          & 0.1572          & 0.1880          \\
    \midrule
\multirow{8}{*}{SCE} & \multirow{4}{*}{SASRec}      & 10                   & 0.0679          & 0.0935          & 0.0479          & 0.0562          & 20  & 0.1760          & 0.2514          & 0.1199          & 0.1442          \\
                     &                              & 20                   & 0.0690          & 0.0957          & 0.0491          & 0.0577          & 50  & 0.2041          & 0.2869          & 0.1397          & 0.1664          \\
                     &                              & 50                   & 0.0698          & 0.0968          & 0.0500          & 0.0587          & 100 & 0.2182          & 0.3084          & 0.1496          & 0.1787          \\
                     &                              & 100                  & \textbf{0.0707} & \textbf{0.0970} & \textbf{0.0506} & \textbf{0.0591} & 200 & \textbf{0.2273} & \textbf{0.3186} & \textbf{0.1567} & \textbf{0.1862} \\
            \cmidrule{2-12}
                     & \multirow{4}{*}{SASRec+}     & 10                   & 0.0682          & 0.0936          & 0.0488          & 0.0569          & 20  & 0.2312          & 0.3231          & 0.1602          & 0.1899          \\
                     &                              & 20                   & 0.0679          & \textbf{0.0937} & 0.0491          & 0.0574          & 50  & 0.2319          & 0.3271          & 0.1607          & 0.1915          \\
                     &                              & 50                   & 0.0684          & 0.0934          & 0.0492          & 0.0572          & 100 & \textbf{0.2342} & \textbf{0.3284} & \textbf{0.1620} & \textbf{0.1924} \\
                     &                              & 100                  & \textbf{0.0687} & 0.0935          & \textbf{0.0496} & \textbf{0.0576} & 200 & 0.2331          & 0.3271          & 0.1617          & 0.1921         \\
    \bottomrule
\end{tabular}
\end{table*}

\end{document}